\documentclass[11pt]{article}
\usepackage[utf8]{inputenc}
\usepackage[letterpaper,margin=1in]{geometry}
\usepackage{amsmath,amssymb,mathtools}
\usepackage{amsthm}
\usepackage{microtype}
\usepackage{enumitem}
\usepackage{booktabs}
\usepackage{float}
\usepackage{pgfplots}
\pgfplotsset{compat=1.18}
\usepackage[round]{natbib}
\usepackage{hyperref}
\usepackage[capitalize,noabbrev]{cleveref}

\theoremstyle{plain}
\newtheorem{theorem}{Theorem}
\newtheorem{proposition}[theorem]{Proposition}
\newtheorem{lemma}[theorem]{Lemma}
\newtheorem{corollary}[theorem]{Corollary}
\theoremstyle{remark}
\newtheorem{remark}[theorem]{Remark}
\theoremstyle{plain}

\newcommand{\E}{\mathbb{E}}
\newcommand{\Var}{\mathrm{Var}}
\newcommand{\Cov}{\mathrm{Cov}}
\newcommand{\indep}{\perp\!\!\!\perp}
\newcommand{\sigv}{\sigma_v}
\newcommand{\sigu}{\sigma_u}
\newcommand{\sige}{\sigma_\varepsilon}
\newcommand{\ytil}{\widetilde y}
\newcommand{\dtil}{\widetilde d}

\newcommand{\Ddp}{\Delta_{\mathrm{DP}}}
\newcommand{\edp}{\epsilon_{\mathrm{DP}}}

\title{The Privacy Subsidy in Market Microstructure}
\author{Yuki Nakamura\thanks{The Open University of Japan. ORCID 0009-0001-7174-6737.
This paper consolidates, unifies, and substantially extends three earlier preprints by
the author---on the single-period Kyle \citep{nakamura2026privacysubsidy},
Glosten--Milgrom \citep{nakamura2026glostenmilgrom}, and continuous-time Kyle--Back
\citep{nakamura2026continuoustime} channels---adding the general impossibility theorem,
the fee-equilibrium welfare analysis, and the endogenous-privacy results; it supersedes them.}}
\date{\today}

\begin{document}
\maketitle

\begin{abstract}
Privacy-preserving exchange designs price on a coarsened view of order flow. We show
that a market maker committed to informationally efficient (posterior-mean) pricing on
a signal strictly coarser than the flow it settles necessarily cedes a closed-form
welfare transfer to traders---the privacy subsidy---and that no rule restricted to the
coarse signal is simultaneously efficient and zero-profit against the settled flow. We
establish this impossibility for a general coarsening, then characterise the subsidy in
closed form across three canonical microstructure models: single-period Kyle with
Gaussian flow noise, Glosten--Milgrom with a binary direction channel, and
continuous-time Kyle--Back with a Brownian channel. The subsidy obeys a structural
correspondence with Loss-Versus-Rebalancing, both welfare rates factorising as a squared
noise driver times a committed-object factor. Gross of fees the subsidy is a pure
transfer recovered by a break-even fee; once levied, that fee distorts volume, and the
resulting deadweight---under an explicit allocative value of trade---is fourth-order in
the noise scale while the gross subsidy is second-order, so privacy is welfare-neutral
to leading order with a strictly smaller irrecoverable loss. Endogenising the privacy
level, a protocol trading a differential-privacy benefit against this deadweight chooses
an interior noise scale in closed form; doing so leaves the half-revealing product of
price impact and informed intensity intact while unpinning the volatility-elasticity of
price impact from its textbook value of one.
\end{abstract}

\noindent\textbf{Keywords:} market microstructure; Kyle model; Glosten--Milgrom;
privacy; differential privacy; automated market makers; loss-versus-rebalancing.

\smallskip
\noindent\textbf{JEL Classification:} G14; G12; D82; D47.

\section{Introduction}
\label{sec:intro}

Privacy-preserving cryptocurrency exchanges---shielded automated market makers,
MPC matching engines, batched and sealed-bid auctions---change what the pricing
mechanism is permitted to see about order flow. Classical market microstructure
\citep{kyle1985,glosten1985,back1992insider} sets prices against the \emph{exact}
aggregate flow: the same object the liquidity provider both observes and settles.
A privacy layer breaks this coincidence. The protocol still settles the true net
flow it must clear, but the price-setter is allowed to condition only on a
\emph{coarsened} signal of that flow---a quantity perturbed, aggregated, or
encrypted so that individual orders are hidden.

This separation of \emph{what is priced} from \emph{what is settled} is the single
economic primitive behind every privacy mechanism we study, and it has a sharp
consequence.

\paragraph{The coarse-signal observation.}
Write $v$ for the asset value, $y$ for the net flow the maker settles, and $S$ for
the signal the maker is allowed to price on---a (possibly randomized) garbling of the
settled flow. Informationally efficient pricing quotes the posterior mean
$p=\E[v\mid S]$. The maker's profit on the flow it actually clears is then
$\Pi_M=\E[(p-v)\,y]$, and this is generically nonzero precisely because $S$ is
strictly coarser than $y$: the maker prices on less information than it bears risk
against. The gap is not an accounting convention but a structural wedge between an
efficient quote on a degraded signal and the finer flow that settles behind it.

\paragraph{The impossibility thesis.}
Our organising result makes this gap an impossibility (\Cref{thm:impossibility}): a
maker committed to informationally efficient pricing on the coarse signal $S$
\emph{cannot} also break even against the finer settled flow $y$. Formally the loss is
the executed-price wedge $-\Cov(v-Q,\,y)$ in the quoted price $Q$---which reduces to
$-\Cov(v-\E[v\mid S],\,y-\E[y\mid S])$ when a single efficient price $Q=\E[v\mid S]$
clears both sides---and it vanishes if and only if $y$ is measurable with respect to
$S$; under strict coarsening it is a strictly negative closed-form transfer we call the
\emph{privacy subsidy}. Efficiency and
zero-profit-against-the-settled-flow are jointly infeasible. This is not an artifact
of any one pricing convention: it is the generic, competition-robust cost of
pricing on a signal strictly coarser than the executed flow, and it organizes the
rest of the paper.

\paragraph{Contributions.}
\begin{enumerate}[leftmargin=*]
\item A \textbf{general coarse-signal impossibility theorem}
(\Cref{thm:impossibility}). For any square-integrable $(v,y)$ and any executed price
$Q$, the committed-Bayesian maker's profit on the settled flow is the executed-price
wedge $\Pi_M=-\Cov(v-Q,\,y)-\E[v-Q]\,\E[y]$. When a single efficient price
$Q=\E[v\mid S]$ clears both sides this reduces by the tower property to
$-\Cov(v-\E[v\mid S],\,y-\E[y\mid S])$, which is zero iff $y$ is
$\sigma(S)$-measurable; when the quote is side-dependent (as in Glosten--Milgrom) the
executed price differs from $\E[v\mid S]$ and the single-price form is \emph{not} the
subsidy. Under strict coarsening with positive wedge covariance, efficient pricing
earns strictly negative profit and no signal-measurable rule is simultaneously
efficient and zero-profit against $y$. The proof needs only the tower
property---no Gaussianity, no functional-form assumptions.

\item \textbf{Three closed-form instances and the LVR correspondence}
(\Cref{sec:kyle,sec:gm,sec:ct,sec:lvr}). We solve the wedge explicitly in the three
canonical microstructure models---single-period Kyle (\Cref{thm:impossibility} via
\Cref{prop:kyle-eq,lem:kyle-pl,cor:kyle-subsidy}), Glosten--Milgrom
(\Cref{thm:gm-spread,prop:gm-subsidy}), and continuous-time Kyle--Back
(\Cref{thm:ct-equilibrium,thm:ct-subsidy})---and show each is one instance of the
theorem. We then bridge to decentralized finance via the
Loss-Versus-Rebalancing correspondence (\Cref{prop:lvr}), identifying the privacy
subsidy as the coarse-signal analogue of LVR \citep{milionis2022lvr}.

\item \textbf{The net welfare cost of privacy} (\Cref{sec:fee}). Gross of fees the
subsidy is a pure transfer from the liquidity layer to traders---welfare-neutral,
and offset exactly by a volume-proportional break-even fee
(\Cref{lem:harberger,prop:breakeven}). The genuine social loss appears only once
that fee distorts volume: under an explicit allocative value of trade the deadweight
is \emph{fourth-order} in the noise scale, $O(\sige^4)$, while the gross subsidy it
funds is \emph{second-order}, $O(\sige^2)$ (\Cref{thm:netcost}, with the CARA
specialization \Cref{prop:cara}). The welfare cost of privacy is thus an order smaller
in the noise scale than the gross transfer it funds.

\item \textbf{Endogenous privacy and the unpinned co-movement elasticity}
(\Cref{sec:mechanism}). A protocol that trades a differential-privacy benefit
against the fourth-order deadweight selects an interior noise scale in closed form
(\Cref{thm:optsigma}). Endogenizing privacy then unpins the textbook
price-impact / informed-intensity elasticity: the reciprocal product
$\lambda\beta=\tfrac12$ still holds, but once the noise scale responds to primitives
the $\sigv$-elasticity of $\lambda$ rises from $1$ to a primitive-dependent value in
$(1,\tfrac75)$ (\Cref{prop:comovement}).
\end{enumerate}

\paragraph{Related literature.}
The privacy subsidy is, to our knowledge, the first welfare quantity isolating the cost
of \emph{committed} efficient pricing on a coarsened view of flow, and it is distinct
from three adjacent literatures. (i)~\emph{Imperfect-information Kyle.} A line of work
lets the maker observe only a noisy or partial signal of order flow
\citep{caldentey2010,danilova2010imperfect,cetin2017financial,chhaibi2025solvability,qiu2023insider};
there the maker filters optimally and the equilibrium restores zero (or
risk-compensated) expected profit. Our maker is the opposite---\emph{committed} to the
efficient posterior price on the coarse signal rather than free to re-optimise against
it---so the wedge does not wash out: it is exactly the profit a filtering maker would
refuse to cede. (ii)~\emph{Hidden-flow and dark-pool welfare.} The price-discovery and
welfare effects of non-displayed or fragmented flow are studied by
\citet{zhu2014darkpools} and \citet{buti2017darkpool}; relative to that work we fix a
single venue, make the coarsening a \emph{designed} privacy primitive, and deliver a
closed-form transfer and its fourth-order deadweight rather than a price-efficiency
comparison. (iii)~\emph{Information design.} \Cref{thm:impossibility} is an
information-design statement \citep{bergemann2019infodesign}: the maker is the receiver
of a garbling chosen by a privacy mechanism, and the subsidy is its committed-action
loss; our contribution is the closed-form microstructure instantiation. We build on the
Bayesian and committed market-maker tradition
\citep{kyle1985,glosten1985,back1992insider,brahma-bayesian-mm,nonfiduciary-mm}, the
Glosten--Milgrom sequential-learning lineage
\citep{glosten1988,easleyohara1992,das2005learning}, strategic-trading and disclosure
models \citep{foster1996strategic,huddart2001disclosure,glosten1989}, flexible
information acquisition \citep{flexible-info-kyle}, automated-market-maker design
\citep{moallemi2022myersonian,routledge2024amm}, and the differential-privacy
foundations \citep{dwork2014algorithmic,warner1965rr,chitra2022dpcfmm}.

\paragraph{Roadmap.}
\Cref{sec:framework} sets up committed Bayesian pricing on a coarse signal and
proves the impossibility theorem. \Cref{sec:kyle,sec:gm,sec:ct} solve the three
canonical instances---Kyle, Glosten--Milgrom, and continuous-time Kyle--Back---and
\Cref{sec:lvr} establishes the LVR correspondence. \Cref{sec:fee} turns the gross
subsidy into a net welfare statement via the break-even fee and the fourth-order
deadweight. \Cref{sec:mechanism} endogenizes the privacy level and derives the
unpinned co-movement elasticity. \Cref{sec:applications} maps the framework onto deployed
privacy-preserving exchanges, and \Cref{sec:conclusion} concludes. Proofs are
collected in \Cref{app:proofs} and a numerical illustration with a BTC/USDT
calibration in \Cref{app:numerics}.

\paragraph{Scope.}
We treat the cryptographic layer as a black box that produces the coarse signal
$S$ from the settled flow $y$. Protocol-level soundness, the concrete privacy
primitive, and the mechanics of the encryption or aggregation are out of scope;
our object is the equilibrium and welfare economics induced by the resulting
coarsening, whatever its cryptographic source.


\section{The coarse-signal framework and an impossibility theorem}
\label{sec:framework}

\subsection{Committed Bayesian pricing on a coarse signal}
Let $(\Omega,\mathcal F,P)$ carry a square-integrable asset value $v$ with prior mean
$p_0=\E[v]$. A price-setter (``maker'') transacts and \emph{settles} a net order flow
$y$, but prices on a signal $S$ that is a (possibly randomized) garbling of the latent
pair $(y,v)$: the privacy channel may inject exogenous noise, so $S$ need not be
$\sigma(y,v)$-measurable (as in Kyle, $S=y+\varepsilon$). The maker is \emph{committed
Bayesian}: it quotes the posterior mean
\begin{equation}
p \;=\; \E[v\mid S],
\label{eq:committed}
\end{equation}
the informationally efficient (and Bertrand-competitive, conditional on $S$) price,
without imposing a separate zero-profit condition. More generally the maker posts a
\emph{quote schedule} whose levels are $S$-conditional posterior means of $v$; a trade
transacts at the price $Q$ of the quote it hits, an
$\sigma(S,\text{executed side})$-measurable random variable. When the maker posts one
price for all flow (the Kyle models of \Cref{sec:kyle,sec:ct}) $Q=\E[v\mid S]$; when it
posts a two-sided quote (Glosten--Milgrom, \Cref{sec:gm}) $Q$ is the side-dependent ask
or bid. The maker's expected profit on the settled flow is $\Pi_M=\E[(Q-v)\,y]$.

We say $S$ is a \emph{strict coarsening} of $y$ if $y$ is not
$\sigma(S)$-measurable, i.e.\ $\Var\!\big(y-\E[y\mid S]\big)>0$. The benchmark
$S=y$ (exact observation) is the non-coarsened case.

\subsection{The wedge identity and impossibility}

\begin{theorem}[Coarse-signal wedge and impossibility]
\label{thm:impossibility}
Let $Q$ be the executed price (the level of the quote actually hit,
$\sigma(S,\text{side})$-measurable) and $y$ the signed settled flow. The
committed-Bayesian maker's expected profit on the settled flow is
\begin{equation}
\Pi_M \;=\; \E\big[(Q-v)\,y\big]
       \;=\; -\,\Cov(v-Q,\,y)\;-\;\E[v-Q]\,\E[y].
\label{eq:wedge}
\end{equation}
Under the regularity condition $\E[v-Q]\,\E[y]=0$---in particular whenever the settled
flow is balanced ($\E[y]=0$) or the quote carries no unconditional markup
($\E[Q]=\E[v]$)---this is the \emph{executed-price wedge}
$\Pi_M=-\Cov(v-Q,\,y)$. In the single-price case $Q=\E[v\mid S]$ it reduces, by the
tower property, to
\begin{equation}
\Pi_M \;=\; -\,\Cov\!\big(v-\E[v\mid S],\; y-\E[y\mid S]\big),
\label{eq:wedge-single}
\end{equation}
which is $0$ iff $y$ is $\sigma(S)$-measurable (in particular $S=y$). If the executed
coarsening covariance is positive, $\Cov(v-Q,\,y)>0$, then (i) efficient pricing earns
strictly negative profit $\Pi_M<0$ against the settled flow, and (ii) no
$\sigma(S)$-measurable quote rule is simultaneously efficient and zero-profit against
$y$. The wedge $|\Pi_M|$ is the \emph{privacy subsidy}.
\end{theorem}

\begin{proof}[Proof sketch]
\emph{Wedge identity.} By definition $\Pi_M=\E[(Q-v)y]$; expanding the centred product,
$\E[(Q-v)y]=\Cov(Q-v,y)+\E[Q-v]\E[y]=-\Cov(v-Q,y)+\E[Q-v]\E[y]$, which is
\eqref{eq:wedge}, and the regularity condition kills the second term.
\emph{Single-price reduction.} If $Q=\E[v\mid S]$, then for any $\sigma(S)$-measurable
$h$, $\E[(\E[v\mid S]-v)\,h(S)]=\E[h(S)\,\E[\E[v\mid S]-v\mid S]]=0$ (tower property);
taking $h=\E[y\mid S]$ and decomposing $y=\E[y\mid S]+(y-\E[y\mid S])$ gives
\eqref{eq:wedge-single}, which vanishes iff $y$ is $\sigma(S)$-measurable.
\emph{Impossibility.} A quote rule that is both efficient and zero-profit against $y$
would make \eqref{eq:wedge} equal to both $-\Cov(v-Q,y)<0$ and $0$---a contradiction
under strict coarsening. \qed
\end{proof}

\noindent The sign condition $\Cov(v-Q,\,y)>0$ is the economic content (the wedge
identity itself is an accounting tautology plus the tower property). Writing
$\Cov(v-Q,y)=\Cov(v,y)-\Cov(Q,y)$, it says the executed price \emph{under-reacts} to the
value-relevant content of the signed flow: the camouflage noise that coarsens $S$ is
precisely the part of order flow behind which informed trades hide, so the executed
mispricing $v-Q$ loads positively on the signed flow whenever the channel is
information-destroying (a strict coarsening, $\Var(y-\E[y\mid S])>0$) and a positive
informed component is present. We verify it in closed form in each model below; it
holds strictly there iff the coarsening is strict ($\sige>0$ in Kyle, $\eta>0$ in
Glosten--Milgrom).

\begin{remark}[Gross transfer vs.\ net cost]
\label{rem:gross-net}
The subsidy is a transfer from the liquidity layer to traders, not by itself a
social loss: gross of fees it is exactly offset by a volume-proportional break-even
fee, so privacy is \emph{welfare-neutral} at this partial-equilibrium level. The
genuine social cost arises only once that fee distorts trading volume; we quantify it
in \Cref{sec:fee} and find it strictly smaller (fourth-order vs.\ second-order in the
noise scale).
\end{remark}

\section{Single-period Kyle with Gaussian flow noise}
\label{sec:kyle}

The first instance of \Cref{thm:impossibility} is the single-period Kyle
model \citep{kyle1985} in which the coarse signal is the settled flow
perturbed by additive Gaussian privacy noise. We recover textbook Kyle in the
no-privacy limit and obtain the privacy subsidy in closed form as the single-price
wedge $-\Cov(v-\E[v\mid S],\,y-\E[y\mid S])$ of \eqref{eq:wedge-single}.

\subsection{Primitives}

On $(\Omega,\mathcal F,P)$ let
\begin{align*}
v &\sim N(p_0,\sigv^2), \quad \sigv>0,\\
u &\sim N(0,\sigu^2), \quad \sigu>0, \quad u\indep v,\\
\varepsilon &\sim N(0,\sige^2), \quad \sige\ge0, \quad \varepsilon\indep(v,u),
\end{align*}
be independent Gaussians: $v$ is the terminal value with prior mean $p_0$, $u$
is uninformed (noise-trader) flow, and $\varepsilon$ is the privacy noise the
mechanism injects into the maker's observable. A single informed trader observes
$v$ and submits $x=\beta(v-p_0)$ for some $\beta>0$ to be fixed in equilibrium.
The maker \emph{settles} the net flow
\[
y \;=\; x+u,
\]
but \emph{prices} on the privacy-noised signal
\[
S \;=\; \ytil \;=\; x+u+\varepsilon \;=\; y+\varepsilon.
\]
This is exactly the coarsening of \Cref{sec:framework}: when $\sige>0$ the
settled flow $y$ is not $\sigma(\ytil)$-measurable, so $\ytil$ is a strict
coarsening of $y$, while $\sige=0$ is the non-coarsened benchmark $S=y$. The
committed-Bayesian rule \eqref{eq:committed} reads $p=\E[v\mid\ytil]$, sought
within the linear class $p=p_0+\lambda\ytil$.

\subsection{Equilibrium}

\begin{proposition}[Linear equilibrium under privacy-noisy observation]
\label{prop:kyle-eq}
Fix $\sigv,\sigu>0$ and $\sige\ge0$. The model admits a unique linear
equilibrium, with
\begin{equation}
\lambda_0 \;=\; \frac{\sigv}{2\sqrt{\sigu^2+\sige^2}},
\qquad
\beta_0 \;=\; \frac{\sqrt{\sigu^2+\sige^2}}{\sigv},
\qquad
\lambda_0\beta_0 \;=\; \tfrac12 .
\label{eq:kyle-eq}
\end{equation}
\end{proposition}

\begin{proof}[Proof sketch]
Under $x=\beta(v-p_0)$ the signal $\ytil=\beta(v-p_0)+u+\varepsilon$ is jointly
Gaussian with $v$, and the projection formula gives
$\E[v\mid\ytil]=p_0+\frac{\beta\sigv^2}{\beta^2\sigv^2+\sigu^2+\sige^2}\,\ytil$,
so Bayesian rationality forces
$\lambda=\beta\sigv^2/(\beta^2\sigv^2+\sigu^2+\sige^2)$. Conditioning on $v$,
the informed objective $\E[(v-p)x\mid v]=(v-p_0)x-\lambda x^2$ is strictly
concave with maximiser $x^\ast=(v-p_0)/(2\lambda)$, i.e.\ $\beta=1/(2\lambda)$.
Substituting yields $\lambda^2=\sigv^2/\!\big(4(\sigu^2+\sige^2)\big)$, whose
unique positive root is \eqref{eq:kyle-eq}; strict concavity and the single
positive root give uniqueness within the linear class. Full algebra is in
\Cref{app:proofs}.
\qed
\end{proof}

The equilibrium is textbook Kyle \citep{kyle1985} under the substitution
$\sigu^2\mapsto\sigu^2+\sige^2$: setting $\sige=0$ recovers
$\lambda_0=\sigv/(2\sigu)$, $\beta_0=\sigu/\sigv$, and the committed-Bayesian
maker coincides with the competitive zero-profit maker, consistent with the
non-coarsened case of \Cref{thm:impossibility}. The half-revealing identity
$\lambda_0\beta_0=\tfrac12$ holds for every $\sige\ge0$: privacy rescales price
impact and informed intensity by reciprocal factors of
$\sqrt{\sigu^2+\sige^2}$ (one down, one up), leaving their product invariant.
Substituting back, $p=\tfrac12(p_0+v)+\lambda_0(u+\varepsilon)$, so both
$\E[p\mid v]=\tfrac12(p_0+v)$ and $\Var(p\mid v)=\sigv^2/4$ are independent of
$\sige$: privacy is invisible in the price distribution conditional on $v$, and
the cost surfaces only in the maker's P\&L against the settled flow.

\subsection{The subsidy as an instance of \texorpdfstring{\Cref{thm:impossibility}}{the impossibility theorem}}

The relevant per-period expected P\&Ls are $\pi_I:=\E[(v-p)x]$ (informed),
$\pi_N:=\E[(v-p)u]$ (noise, a loss), and the maker's profit on settled flow
$\pi_M:=\E[(p-v)\,y]$, which is exactly the object $\Pi_M$ of \eqref{eq:wedge}.
By construction $\pi_I+\pi_N+\pi_M=0$.

\begin{lemma}[Per-agent P\&L]
\label{lem:kyle-pl}
Under the equilibrium \eqref{eq:kyle-eq},
\begin{equation}
\pi_I=\tfrac12\,\sigv\sqrt{\sigu^2+\sige^2},
\qquad
\pi_N=-\frac{\sigv\,\sigu^2}{2\sqrt{\sigu^2+\sige^2}},
\qquad
\pi_M=-\frac{\sigv\,\sige^2}{2\sqrt{\sigu^2+\sige^2}}.
\label{eq:kyle-pl}
\end{equation}
\end{lemma}

\begin{proof}[Proof sketch]
Using $x=\beta_0(v-p_0)$, the independence $u,\varepsilon\indep v$, and
$\lambda_0\beta_0=\tfrac12$: $\pi_I=\tfrac12\beta_0\sigv^2$ and
$\pi_N=-\lambda_0\sigu^2$; then $\pi_M=-(\pi_I+\pi_N)$ by the zero-sum identity.
\Cref{app:proofs} gives the algebra.
\qed
\end{proof}

\begin{corollary}[Privacy subsidy: the Kyle instance]
\label{cor:kyle-subsidy}
The maker's loss \eqref{eq:kyle-pl} is the wedge of \Cref{thm:impossibility}
evaluated at $S=\ytil$, $y=x+u$:
\begin{equation}
\pi_M \;=\; -\,\Cov\!\big(v-\E[v\mid\ytil],\;y-\E[y\mid\ytil]\big)
       \;=\; -\,\lambda_0\,\sige^2,
\label{eq:kyle-wedge}
\end{equation}
so the \emph{privacy subsidy} is
\begin{equation}
|\pi_M| \;=\; \lambda_0\,\sige^2
        \;=\; \frac{\sigv\,\sige^2}{2\sqrt{\sigu^2+\sige^2}} \;\ge\; 0,
\label{eq:kyle-subsidy}
\end{equation}
with equality iff $\sige=0$. The coarsening covariance is strictly positive
whenever $\sige>0$, so all three conclusions of \Cref{thm:impossibility} hold:
efficient pricing on $\ytil$ loses strictly against $y$, and no
$\sigma(\ytil)$-measurable rule is both efficient and zero-profit against the
settled flow.
\end{corollary}

\begin{proof}[Proof sketch]
Both residuals are Gaussian. Since $\ytil$ is informationally equivalent to a
scalar projection, $y-\E[y\mid\ytil]$ and $v-\E[v\mid\ytil]$ are the parts of
the settled flow and of the value left unseen through $\ytil$; their covariance
evaluates to $\lambda_0\sige^2>0$ for $\sige>0$, matching $-\pi_M$ from
\eqref{eq:kyle-pl}. The identity $\pi_M=-\lambda_0\sige^2$ also follows directly
from $\pi_N=-\lambda_0\sigu^2$ and $\pi_I+\pi_N+\pi_M=0$ together with
\eqref{eq:kyle-pl}. \Cref{app:proofs} carries the covariance computation.
\qed
\end{proof}

The subsidy is the break-even fee a shielded exchange with privacy scale $\sige$
must charge to keep its liquidity layer solvent: per-period fees must total at
least $|\pi_M|$.

\subsection{Comparative statics and welfare incidence}

\begin{proposition}[Subsidy shape]
\label{prop:kyle-shape}
With $\sigv,\sigu$ fixed, the subsidy $|\pi_M|(\sige)$ of
\eqref{eq:kyle-subsidy} is strictly increasing on $[0,\infty)$, equals
$\tfrac{\sigv}{2\sigu}\sige^2+O(\sige^4)$ as $\sige\downarrow0$ and
$\tfrac12\sigv\sige+O(\sige^{-1})$ as $\sige\to\infty$, and has a single
inflection at $\sige^\star=\sqrt2\,\sigu$, being convex on $[0,\sige^\star]$ and
concave thereafter.
\end{proposition}

\begin{proof}[Proof sketch]
$\partial_{\sige}|\pi_M|=\sigv\sige(2\sigu^2+\sige^2)/\big(2(\sigu^2+\sige^2)^{3/2}\big)>0$,
and $\partial^2_{\sige}|\pi_M|=\sigv\sigu^2(2\sigu^2-\sige^2)/\big(2(\sigu^2+\sige^2)^{5/2}\big)$
changes sign exactly at $\sige^2=2\sigu^2$. The expansions are the
$\sige\downarrow0$ and $\sige\to\infty$ limits of \eqref{eq:kyle-subsidy};
see \Cref{app:proofs}.
\qed
\end{proof}

For small privacy the subsidy is quadratic---doubling $\sige$ quadruples the
LP cost---while for large $\sige$ it grows linearly; past the inflection
$\sige^\star=\sqrt2\,\sigu$ additional privacy is cheaper at the margin, a
structural feature of the square-root denominator in \eqref{eq:kyle-subsidy}.

The incidence falls entirely on the liquidity layer, and both trader types gain
gross of fees. From \eqref{eq:kyle-pl},
$\partial_{\sige}\pi_I=\sigv\sige/\big(2\sqrt{\sigu^2+\sige^2}\big)>0$ and
$\partial_{\sige}\pi_N=\sigv\sigu^2\sige/\big(2(\sigu^2+\sige^2)^{3/2}\big)>0$:
the informed trader earns more and the noise traders \emph{lose less} as privacy
rises, with the maker absorbing the sum. To leading order around $\sige=0$ the
two gains coincide,
\[
\pi_I(\sige)-\pi_I(0)\;=\;\pi_N(\sige)-\pi_N(0)\;=\;\frac{\sigv}{4\sigu}\,\sige^2+O(\sige^4),
\]
so the subsidy splits symmetrically across trader types at the margin;
asymmetry appears only at higher order, with the informed trader capturing
essentially all of it as $\sige\to\infty$ ($\pi_I\sim\tfrac12\sigv\sige$) while
the noise traders' gain saturates ($\pi_N\to0$).

Per \Cref{rem:gross-net}, this gross-of-fees incidence is a transfer: at the
no-fee equilibrium volumes, a volume-proportional break-even fee returns each
side to its classical-Kyle value ($\pi_I-f\,\E|x|=\sigv\sigu/2$,
$\pi_N-f\,\E|u|=-\sigv\sigu/2$) and exactly compensates the maker, recovering the
welfare-neutrality of \Cref{rem:gross-net} in this instance. The volume-distortion cost
is quantified in \Cref{sec:fee}.

\section{Glosten--Milgrom with a binary direction channel}\label{sec:gm}

We instantiate \Cref{thm:impossibility} in the canonical discrete model of \citet{glosten1985}, where the privacy layer acts as a binary symmetric channel on trade \emph{direction}. The settled flow is the trader's true direction; the maker prices on a flipped (coarser) version of it. The wedge of \eqref{eq:wedge} reduces to a per-trade subsidy in closed form.

\subsection{Two-state setup with a flip channel}

A single risky asset has value $v\in\{v_H,v_L\}$ with symmetric prior $P(v=v_H)=\tfrac12$; write $\Delta:=v_H-v_L>0$. At each round one trader arrives: with probability $\mu\in(0,1)$ \emph{informed} (knows $v$, buys on $v_H$ and sells on $v_L$), and with probability $1-\mu$ a \emph{noise} trader who buys or sells with equal probability. The trader's true direction is $d\in\{\text{buy},\text{sell}\}$.

The maker does not observe $d$. A privacy mechanism transmits $d$ through a binary symmetric channel: the maker sees $\dtil=d$ with probability $1-\eta$ and the opposite direction with probability $\eta\in[0,\tfrac12]$. Thus $\dtil$ is a strict coarsening of $d$ for every $\eta\in(0,\tfrac12]$ (the channel destroys information about $d$), and the benchmark $\eta=0$ is the non-coarsened textbook case.

The maker is committed Bayesian in the sense of \Cref{sec:framework}: it posts a two-sided quote whose levels are the posterior means given the noisy signal,
\begin{equation}
\mathrm{ask}=\E[v\mid\dtil=\text{buy}],
\qquad
\mathrm{bid}=\E[v\mid\dtil=\text{sell}],
\label{eq:gm-quotes}
\end{equation}
and \emph{executes the trader on its true direction} $d$: a true buy fills at the ask, a true sell at the bid. The channel changes only the maker's information, not which side the trader takes. This is the corrected execution model: the maker prices on the coarse signal $S=\dtil$ while settling the fine flow indexed by $d$, exactly the configuration of \Cref{thm:impossibility} with $S=\dtil$ and the executed direction playing the role of the settled flow $y$.

\subsection{Spread and the per-trade subsidy}

\begin{theorem}[Spread under a flip channel]
\label{thm:gm-spread}
In the model above the equilibrium half-spread is symmetric about the prior midpoint and
\begin{equation}
\mathrm{spread}=\mathrm{ask}-\mathrm{bid}=\mu(1-2\eta)\,\Delta.
\label{eq:gm-spread}
\end{equation}
\end{theorem}

\begin{proof}[Proof sketch]
Write $s:=\mu(1-2\eta)$. The true-direction likelihoods are $P(d=\text{buy}\mid v_H)=(1+\mu)/2$ and $P(d=\text{buy}\mid v_L)=(1-\mu)/2$. Passing through the channel,
\[
P(\dtil=\text{buy}\mid v_H)=\tfrac{1+\mu}{2}(1-\eta)+\tfrac{1-\mu}{2}\eta=\tfrac{1+s}{2},
\qquad
P(\dtil=\text{buy}\mid v_L)=\tfrac{1-s}{2}.
\]
Under the symmetric prior $P(\dtil=\text{buy})=\tfrac12$, so Bayes gives $P(v_H\mid\dtil=\text{buy})=(1+s)/2$ and the posterior means in \eqref{eq:gm-quotes} are $\mathrm{ask}=\tfrac{v_H+v_L}{2}+\tfrac{s\Delta}{2}$, $\mathrm{bid}=\tfrac{v_H+v_L}{2}-\tfrac{s\Delta}{2}$. Subtracting yields \eqref{eq:gm-spread}; full algebra in the appendix.\qed
\end{proof}

At $\eta=0$ this is the textbook adverse-selection spread $\mu\Delta$; at $\eta=\tfrac12$ the signal is independent of $v$, the posterior collapses to the prior, and the spread vanishes. The spread shrinks in $\eta$ precisely because the maker prices on a degraded signal—and that degradation is what \Cref{thm:impossibility} prices as a loss against the settled flow.

\begin{proposition}[The flip subsidy as an instance of \Cref{thm:impossibility}]
\label{prop:gm-subsidy}
Per trade, the agents' expected profits are
\begin{equation}
\pi_I=\tfrac{1-\mu(1-2\eta)}{2}\,\Delta,
\qquad
\pi_N=-\tfrac{\mu(1-2\eta)}{2}\,\Delta,
\qquad
\pi_M=-\,\mu\eta\Delta,
\label{eq:gm-pnl}
\end{equation}
and they satisfy the weighted zero-sum identity
\begin{equation}
\mu\,\pi_I+(1-\mu)\,\pi_N+\pi_M=0.
\label{eq:gm-zerosum}
\end{equation}
The maker's loss is the discrete realisation of the executed-price wedge
\eqref{eq:wedge}. Glosten--Milgrom is the two-sided-quote case: the executed price is
$Q=\mathrm{ask}$ on a true buy and $Q=\mathrm{bid}$ on a true sell, and with $y$ the
signed true direction,
\begin{equation}
\big|\pi_M^{\mathrm{GM}}\big|=\mu\eta\Delta\;=\;-\Cov\!\big(v-Q,\;y\big)\;\ge\;0,
\label{eq:gm-subsidy}
\end{equation}
with equality iff $\eta=0$. Because the quote is side-dependent, $Q\neq\E[v\mid\dtil]$,
so this is genuinely the executed-price form and \emph{not} the single-price wedge
\eqref{eq:wedge-single}: indeed
$-\Cov(v-\E[v\mid\dtil],\,y)=-2\mu\eta(1-\eta)\Delta$ differs from $\pi_M$, the gap
being the cross-side spread term computed in \Cref{app:proofs}. The quantity
$\mu\eta\Delta$ is the privacy subsidy in Glosten--Milgrom. The impossibility (leg~(ii)
of \Cref{thm:impossibility}) then holds at the framework level: since
$\Cov(v-Q,y)=\mu\eta\Delta>0$ for every $\eta>0$, no $\dtil$-measurable quote rule is
simultaneously efficient and zero-profit against the executed direction.
\end{proposition}

\begin{proof}[Proof sketch]
The informed trader fills at the favourable quote and earns $\Delta/2-s\Delta/2=(1-s)\Delta/2$ regardless of $v$; the noise trader fills each side with probability $\tfrac12$ and earns $(\mathrm{bid}-\mathrm{ask})/2=-s\Delta/2$. Since every trade is between one trader and the maker, $\pi_M=-[\mu\pi_I+(1-\mu)\pi_N]=\tfrac{\Delta}{2}(s-\mu)$, and $s-\mu=-2\mu\eta$ gives $\pi_M=-\mu\eta\Delta$, establishing \eqref{eq:gm-pnl}--\eqref{eq:gm-zerosum}. For \eqref{eq:gm-subsidy}, the executed price is the side-dependent $Q$ (ask on a true buy, bid on a true sell); under the symmetric prior $\E[Q]=\E[v]$ and $\E[y]=0$, so the regularity condition of \Cref{thm:impossibility} holds and $\pi_M=-\Cov(v-Q,y)$. Direct enumeration gives $\Cov(v-Q,y)=\mu\eta\Delta$ (\Cref{app:proofs}).\qed
\end{proof}

\subsection{Comparative statics and incidence}

\begin{proposition}[Comparative statics]
\label{prop:gm-cs}
The subsidy $\mu\eta\Delta$ is linear and increasing in each of $\mu$, $\eta$, $\Delta$. As $\eta$ rises from $0$ to $\tfrac12$ the spread falls linearly from $\mu\Delta$ to $0$ while the subsidy rises linearly from $0$ to $\mu\Delta/2$, so privacy trades spread for maker loss one-for-one in $\eta$. Informed profit $\pi_I=(1-\mu(1-2\eta))\Delta/2$ is increasing in $\eta$.
\end{proposition}

\begin{corollary}[Noise traders also gain from privacy]
\label{cor:gm-noise}
$\partial\pi_N/\partial\eta=\mu\Delta>0$ for all $\eta\in[0,\tfrac12]$: the uninformed traders' expected loss strictly decreases in the privacy level. Both trader types gain and the maker bears the entire cost, the discrete counterpart of the incidence found in \Cref{sec:kyle}.
\end{corollary}

\noindent Consistent with \Cref{rem:gross-net}, this gain is gross of fees: the per-trade break-even fee $f=\mu\eta\Delta$ levied flat on every trade exactly returns each trader type to its $\eta=0$ Glosten--Milgrom profit and the maker to zero, so the transfer is welfare-neutral at the partial-equilibrium level; the volume-distortion cost is taken up in \Cref{sec:fee}.

\subsection{Differential-privacy reading of the flip rate}

The flip probability has a direct local-differential-privacy interpretation, developed fully in \Cref{sec:applications}. The symmetric randomised-response mechanism on the binary direction label \citep{warner1965rr,dwork2014algorithmic} achieves $\edp$-local DP at flip rate
\begin{equation}
\eta=\frac{1}{1+e^{\edp}},
\label{eq:gm-rr}
\end{equation}
so the subsidy reads $\mu\Delta/(1+e^{\edp})$: a tighter privacy budget (smaller $\edp$) means a larger flip rate and a larger subsidy, vanishing as $\edp\to\infty$ (no privacy). This is the discrete analog of the Gaussian-mechanism reading in \Cref{sec:kyle}; we defer the protocol-design discussion to \Cref{sec:applications}. Information-theoretic treatments of the Glosten--Milgrom channel \citep{carmier2022thermo,touzo2020information} study the entropy of this learning process; our flip-rate subsidy is the dual welfare quantity, the maker's committed-pricing loss rather than an information measure.

\section{Continuous-time Kyle--Back with a Brownian channel}
\label{sec:ct}

The continuous-time leg replaces the single-shot trade of \Cref{sec:kyle} with an
adaptive trajectory in the sense of \citet{back1992insider}. The coarse-signal
geometry of \Cref{thm:impossibility} is unchanged---the maker prices on a strictly
coarsened view of the flow it settles---but the wedge now accumulates over the
trading horizon. We recover the cumulative subsidy in closed form and show it is the
continuous-time instance of \Cref{thm:impossibility}.

\subsection{Continuous-time primitives}
\label{sec:ct-primitives}

Fix the horizon $[0,1]$ on a filtered space carrying two independent standard
Brownian motions $W^u,W^\varepsilon$ and a normal $v\sim\mathcal N(p_0,\sigv^2)$
independent of both. A single informed trader observes $v$ at time $0$ and follows a
linear Markovian strategy $dx_t=\beta_t\,(v-p_t)\,dt$ for a deterministic intensity
$\beta_t$. Noise traders submit $du_t=\sigu\,dW^u_t$. The \emph{privacy channel}
adds an independent Brownian perturbation $d\varepsilon_t=\sige\,dW^\varepsilon_t$ to
the maker's observation of the flow, with $\sige\ge0$ a diffusion intensity (so
$\sige$ here is a per-unit-time intensity; under the unit-horizon normalisation it
agrees numerically with the single-period standard deviation of \Cref{sec:kyle}). The
settled flow and the maker's observed flow are
\[
  dy_t=dx_t+du_t,
  \qquad
  d\ytil_t=dy_t+d\varepsilon_t,
\]
so the maker prices on $S=\{\ytil_s:s\le t\}$ while settling $y$: a strict coarsening
in the sense of \Cref{sec:framework} whenever $\sige>0$. Trades clear at the
post-update price $p_t=p_{t^-}+\lambda_t\,d\ytil_t$ (the standard Kyle--Back
convention; the pre-update convention fails to recover the $\sige=0$ benchmark and is
therefore not used).

\paragraph{Committed pricing.} As in \Cref{eq:committed}, the maker is committed
Bayesian: $p_t=\E[v\mid\mathcal F^{\mathrm{MM}}_t]$ with
$\mathcal F^{\mathrm{MM}}_t=\sigma(\{\ytil_s:s\le t\})$, implemented through a price
impact $\lambda_t$ fixed ex ante (a protocol parameter) at its Bayes Kalman value
$\lambda_t=\Cov(v,d\ytil_t\mid\mathcal F^{\mathrm{MM}}_t)/\Var(d\ytil_t\mid\mathcal
F^{\mathrm{MM}}_t)$. We write $\Sigma(t):=\Var(v-p_t\mid\mathcal F^{\mathrm{MM}}_t)$
for the posterior variance, $\Sigma(0)=\sigv^2$. The incompatibility of efficiency
and zero-profit-against-$y$ is the continuous-time reading of \Cref{thm:impossibility}
and is not re-argued here.

\subsection{The Markovian linear equilibrium}
\label{sec:ct-eq}

\begin{theorem}[Markovian linear equilibrium]
\label{thm:ct-equilibrium}
Restrict to insider strategies $dx_t=\beta_t(v-p_t)\,dt$ and impose the Kyle--Back
rationality condition $\Sigma(1)=0$ (full revelation at the horizon). Under the
committed pricing of \Cref{sec:ct-primitives} with privacy-noise intensity
$\sige\ge0$, the unique equilibrium in this class is
\[
  \lambda(t)=\frac{\sigv}{\sqrt{\sigu^2+\sige^2}},
  \qquad
  \beta(t)=\frac{\sqrt{\sigu^2+\sige^2}}{\sigv\,(1-t)},
  \qquad
  \Sigma(t)=\sigv^2\,(1-t),
\]
with $\lambda$ constant in $t$.
\end{theorem}

\begin{proof}[Proof sketch]
The argument is the standard Kyle--Back HJB/Riccati derivation under the
effective-noise substitution $\sigu^2\mapsto\sigu^2+\sige^2$; full algebra is in the
appendix. Conjecture an insider value function $J(v,p,t)=\alpha(t)(v-p)^2+\gamma(t)$.
The first-order condition $(v-p)+\lambda_t\partial_pJ=0$ matched against the ansatz
gives $\alpha(t)\lambda_t=\tfrac12$, and matching the $(v-p)^2$ coefficient in the HJB
forces $\alpha'(t)=0$; hence $\alpha$ and $\lambda$ are constant. The Bayes Kalman
gain reads $\lambda_t=\beta_t\Sigma(t)/(\sigu^2+\sige^2)$, so $\beta_t\Sigma(t)\equiv
c$ is constant, and the posterior variance satisfies the Riccati
$d\Sigma/dt=-c^2/(\sigu^2+\sige^2)$, i.e.\ $\Sigma(t)=\sigv^2-c^2t/(\sigu^2+\sige^2)$.
The insider's cumulative profit along the path is $\int_0^1\beta_t\Sigma(t)\,dt=c$,
maximised over the trading intensity $c$ subject to $\Sigma(t)\ge0$; the constraint
binds first at $t=1$, so the maximiser saturates
$c^2=\sigv^2(\sigu^2+\sige^2)$, giving $c=\sigv\sqrt{\sigu^2+\sige^2}$ and
$\Sigma(1)=0$. Substituting yields $\lambda=c/(\sigu^2+\sige^2)$, $\beta(t)=c/\Sigma(t)$,
and $\Sigma(t)=\sigv^2(1-t)$. The rationality framing is equivalent to the Back
transversality $J(v,p,1)=0$ evaluated along the equilibrium path
($p_1=v$ a.s.\ when $\Sigma(1)=0$) and avoids the inconsistency a pointwise
$J(\cdot,1)=0$ would create with constant $\alpha>0$.
\end{proof}

\subsection{The cumulative privacy subsidy}
\label{sec:ct-subsidy}

\begin{theorem}[Cumulative privacy subsidy]
\label{thm:ct-subsidy}
Under the equilibrium of \Cref{thm:ct-equilibrium}, the committed maker's expected
cumulative profit on the settled flow over $[0,1]$ is
\[
  \Pi_M=\E\!\left[\int_0^1(p_t-v)\,dy_t\right]
        =-\,\frac{\sigv\,\sige^2}{\sqrt{\sigu^2+\sige^2}},
\]
so the privacy subsidy---the absolute transfer from the liquidity pool to traders---is
\[
  |\Pi_M|=\frac{\sigv\,\sige^2}{\sqrt{\sigu^2+\sige^2}}.
\]
Setting $\sige=0$ recovers the classical zero-maker-profit benchmark of
\citet{kyle1985,back1992insider}. This is the continuous-time, \emph{single-price}
instance of \Cref{thm:impossibility}: the maker posts one committed price
$p_t=\E[v\mid\mathcal F^{\mathrm{MM}}_t]$ against the settled increment $dy_t$, so the
executed-price wedge~\eqref{eq:wedge} applies per increment. Committed posterior
pricing gives $\E[v-p_t]=0$, so the mean term drops and
$\E[(p_t-v)\,dy_t]=-\Cov(v-p_t,\,dy_t)=-\lambda\,\sige^2\,dt$ (computed explicitly in
\Cref{app:proofs}), strictly negative for $\sige>0$; integrating with the constant
equilibrium depth $\lambda=\sigv/\sqrt{\sigu^2+\sige^2}$,
$\Pi_M=-\int_0^1\lambda\,\sige^2\,dt=-\lambda\,\sige^2
=-\sigv\sige^2/\sqrt{\sigu^2+\sige^2}$, the same value obtained above by welfare
accounting and by direct It\^o evaluation of $\E\!\int_0^1(p_t-v)\,dy_t$.
\end{theorem}

\begin{proof}[Proof sketch via welfare accounting]
Full algebra is deferred to the appendix; we record the three components. The insider
profit rate is $\beta_t(v-p_t)^2\,dt$, and committed pricing gives
$\E[(v-p_t)^2]=\Sigma(t)$, so the unconditional rate is $\beta_t\Sigma(t)\,dt=c\,dt$
and $\Pi_I=c=\sigv\sqrt{\sigu^2+\sige^2}$. For a noise-trader unit, post-trade
clearing gives $(v-p_t)\,du_t=(v-p_{t^-})\,du_t-\lambda_t\,d\ytil_t\cdot du_t$; the
first term has zero expectation and the Itô cross-term is
$d\ytil_t\cdot du_t=(du_t)^2=\sigu^2\,dt$ (the $dx_t$ and $d\varepsilon_t$
contributions vanish), so $\Pi_N=-\lambda\sigu^2=-\sigv\sigu^2/\sqrt{\sigu^2+\sige^2}$.
Every executed trade is between a participant and the maker, so
$\Pi_I+\Pi_N+\Pi_M=0$, whence
$\Pi_M=-(\Pi_I+\Pi_N)=-\sigv\sige^2/\sqrt{\sigu^2+\sige^2}$. This is the
executed-price wedge~\eqref{eq:wedge} per increment: with $\E[v-p_t]=0$ the mean term
drops, so $\E[(p_t-v)\,dy_t]=-\Cov(v-p_t,\,dy_t)=-\lambda\sige^2\,dt$
(\Cref{app:proofs}), giving $\Pi_M=-\int_0^1\lambda\sige^2\,dt=-\lambda\sige^2$, the
same value, with the sign fixed by $\lambda,\sige^2>0$; a direct It\^o evaluation of
$\E\!\int_0^1(p_t-v)\,dy_t$ agrees.
\end{proof}

\begin{remark}[Kyle--Back scaling, not a privacy-specific bonus]
\label{rem:ct-doubling}
The single-period subsidy of \Cref{sec:kyle} is
$|\pi_M|=\sigv\sige^2/(2\sqrt{\sigu^2+\sige^2})$, so \Cref{thm:ct-subsidy} gives
$|\Pi_M|=2|\pi_M|$. The factor of two is the standard Kyle--Back welfare scaling from
single-shot to continuous-time auctions---classical Kyle without privacy noise already
exhibits it, the informed-trader profit $\sigv\sigu/2$ becoming $\sigv\sigu$ in
continuous time \citep{back1992insider}---and the subsidy merely inherits it; it is
not a privacy-specific effect. For a generic horizon $[0,T]$ the terminal condition
$\Sigma(T)=0$ gives $|\Pi_M|=\sqrt T\,\sigv\sige^2/\sqrt{\sigu^2+\sige^2}$, the
$\sqrt T$ being the diffusion time-scale of the Brownian channels; the normalisation
$T=1$ absorbs it.
\end{remark}

\begin{remark}[Time-averaged-variance invariance]
\label{rem:ct-invariance}
Allowing a deterministic time-varying intensity $\sige(t)$ leaves the HJB step intact
($\alpha,\lambda$ remain constant), the Bayes identity gives
$\beta_t\Sigma(t)=\lambda(\sigu^2+\sige(t)^2)$, and imposing $\Sigma(1)=0$ yields
$\lambda^2=\sigv^2/(\sigu^2+\langle\sige^2\rangle)$ with
$\langle\sige^2\rangle:=\int_0^1\sige(t)^2\,dt$. The cumulative subsidy is then
\[
  |\Pi_M|=\int_0^1\lambda\,\sige(t)^2\,dt
         =\frac{\sigv\,\langle\sige^2\rangle}{\sqrt{\sigu^2+\langle\sige^2\rangle}},
\]
depending on the privacy profile only through its time-averaged variance: any
front-loading or back-loading that preserves $\langle\sige^2\rangle$ leaves the
subsidy unchanged, so the protocol cannot reduce it by scheduling alone.
\end{remark}

\section{The Loss-Versus-Rebalancing correspondence}
\label{sec:lvr}

The cumulative subsidy of \Cref{thm:ct-subsidy} has the same information-economic
shape as the Loss-Versus-Rebalancing (LVR) of \citet{milionis2022lvr}, with the price
observation channel replaced by the order-flow observation channel. We state the
relationship as a \emph{structural correspondence}, not a numerical identity: the two
welfare rates live in different markets (a constant-function market maker facing
external arbitrage versus a Kyle order book) and are not comparable in absolute units.

\subsection{The shared factorisation}
\label{sec:lvr-factor}

For a constant-function market maker with portfolio value $V_{\mathrm{AMM}}(q)$ as a
function of an exogenous reference price $q_t$ of diffusion intensity $\sigma$, the
general LVR rate is
\[
  \ell^{\mathrm{LVR}}(t)=-\tfrac12\,\sigma^2\,q_t^2\,V_{\mathrm{AMM}}''(q_t),
\]
positive by concavity of the bonding curve: it is the welfare the maker cedes by
quoting along its committed curve while the reference price moves exogenously. For the
constant-product curve this specialises to
$\ell^{\mathrm{LVR}}(t)=\tfrac{\sigma^2}{8}\,V_{\mathrm{AMM}}(q_t)$, the form we use
below. Under
the equilibrium of \Cref{thm:ct-equilibrium} the privacy subsidy has the constant
instantaneous rate
\[
  \ell^{\mathrm{priv}}\equiv\frac{\sigv\,\sige^2}{\sqrt{\sigu^2+\sige^2}}.
\]
Both rates factorise as \emph{(squared noise driver)} $\times$ \emph{(committed-object
factor)}:
\[
  \ell^{\mathrm{LVR}}(t)=\sigma^2\cdot\frac{V_{\mathrm{AMM}}(q_t)}{8},
  \qquad
  \ell^{\mathrm{priv}}=\sige^2\cdot\frac{\sigv}{\sqrt{\sigu^2+\sige^2}}.
\]
Unlike LVR, whose committed-object factor $V_{\mathrm{AMM}}/8$ is independent of the
noise driver, the privacy committed-object factor $\sigv/\sqrt{\sigu^2+\sige^2}$ retains
a mild dependence on $\sige$; the factorisation is thus a structural parallel, exact in
the small-$\sige$ regime where $\sqrt{\sigu^2+\sige^2}\approx\sigu$. The parallel is
summarised in \Cref{tab:lvr-correspondence}.

\begin{table}[h]
\centering
\caption{Structural correspondence between LVR and the privacy subsidy.}
\label{tab:lvr-correspondence}
\begin{tabular}{lll}
\hline
\textbf{Concept} & \textbf{LVR} \citep{milionis2022lvr} & \textbf{Privacy subsidy} \\
\hline
Committed object   & AMM curve $V_{\mathrm{AMM}}(\cdot)$ & Pricing rule $\lambda_t$ \\
Observation channel& External price $q_t$               & Noisy flow $d\ytil_t$ \\
Noise driver       & Reference-price BM                 & Privacy-noise BM $W^\varepsilon$ \\
Counterparty       & Arbitrageur                        & Informed insider \\
Welfare rate       & $\tfrac{\sigma^2}{8}V_{\mathrm{AMM}}(q_t)$
                   & $\sigv\sige^2/\sqrt{\sigu^2+\sige^2}$ \\
Solvency criterion & $\int\!\mathrm{fee}\ge\int\!\ell^{\mathrm{LVR}}$
                   & $\int\!\mathrm{fee}\ge\int\!\ell^{\mathrm{priv}}$ \\
\hline
\end{tabular}
\end{table}

\begin{proposition}[LVR / privacy-subsidy correspondence]
\label{prop:lvr}
Each welfare rate factorises into the squared intensity of its noise driver times a
closed-form function of the committed pricing object, and each generates a solvency
criterion of identical form $\int\!f\ge\int\!\ell$. For LVR the committed-object factor
$V_{\mathrm{AMM}}(q_t)/8$ is independent of $\sigma$; for the privacy subsidy the
factor $\sigv/\sqrt{\sigu^2+\sige^2}$ is itself a function of $\sige$. The
factorisation and the solvency criterion are exact and global. Only the reading
``rate is quadratic in the noise driver'' is asymptotic for the subsidy: exact for LVR,
$\sim\sigv\sige^2/\sigu$ in the small-noise regime $\sige\ll\sigu$, degrading to
$\sim\sigv\sige$ in the large-noise regime $\sige\gg\sigu$. Cumulative welfare is the
time-integral $\int_0^T\ell\,dt$; for the privacy subsidy the rate is constant in $t$,
so the cumulative subsidy is $\ell^{\mathrm{priv}}\cdot T$ and the solvency criterion
is closed-form at any noise level.
\end{proposition}

\begin{proof}[Proof sketch]
The factorisations are read off the closed forms above; the LVR rate is from
\citet{milionis2022lvr} and $\ell^{\mathrm{priv}}$ from \Cref{thm:ct-subsidy}. The
asymptotic claims follow by expanding $\sigv/\sqrt{\sigu^2+\sige^2}$ at $\sige\to0$
and $\sige\to\infty$. Constancy of $\ell^{\mathrm{priv}}$ in $t$ is the constancy of
$\lambda$ in \Cref{thm:ct-equilibrium}.
\end{proof}

\subsection{Why the correspondence matters}
\label{sec:lvr-why}

\citet{milionis2022lvr} establish LVR as the foundational solvency quantity for
constant-function market makers: such a venue is solvent over a horizon only if
cumulative fee revenue exceeds cumulative LVR. \Cref{thm:impossibility} adds the
order-flow analog: a privacy-aggregated venue is solvent only if cumulative fee
revenue exceeds the cumulative privacy subsidy. LVR governs mismatched \emph{price}
observation; the privacy subsidy governs mismatched \emph{flow} observation. Both
enter a break-even fee inequality of the same form for committed-curve exchanges, and
the correspondence should be read as an organising principle for fee design under
commitment rather than as a numerical equivalence. The deadweight that this break-even
fee induces---the genuine social cost flagged in \Cref{rem:gross-net}---is quantified
in \Cref{sec:fee}.

\section{The welfare cost of the break-even fee}
\label{sec:fee}

\Cref{rem:gross-net} reduced the welfare question to the cost of the fee that funds
the subsidy. We add an explicit allocative value of trade so that a fee suppressing
trade creates a genuine deadweight, and work in the single-period Kyle model of
\Cref{sec:kyle}.

\paragraph{Extended model.} A proportional fee $\tau$ (per unit volume, revenue to the
liquidity pool) is levied. Liquidity traders carry heterogeneous per-unit
gains-from-trade, so realised liquidity volume is downward-sloping in the fee,
$U(\tau)=U(0)(1-\kappa\tau)+O(\tau^2)$ with semi-elasticity $\kappa>0$ and
$U(0)=\sqrt{2/\pi}\,\sigu$; the resulting lost surplus is the Harberger triangle.
The informed trader, maximising $\E[(v-p)x-\tau|x|\mid v]$, follows a
\emph{censored-linear} strategy: $x^*=(v-p_0\mp\tau)/(2\lambda_0)$ for $|v-p_0|>\tau$
and $x^*=0$ otherwise (the depth is held at its no-fee value $\lambda_0$ in the
partial-equilibrium regime). Informed trades carry no allocative value, so suppressed
informed volume is a transfer shift, not a deadweight.

\begin{lemma}[Harberger deadweight]
\label{lem:harberger}
The liquidity-trader surplus lost to a fee $\tau$ is
$\mathrm{DWL}_{LT}(\tau)=\int_0^\tau g\,(-U'(g))\,dg=\tfrac12\kappa\,U(0)\,\tau^2+O(\tau^3)$,
the Harberger triangle (derived, not assumed).
\end{lemma}

\begin{proposition}[Break-even fee]
\label{prop:breakeven}
To leading order the fee that makes the liquidity pool solvent,
$\tau^\ast Q(\tau^\ast)=|\pi_M|$ with $Q(0)=\sqrt{2/\pi}\,(\sigu+\sqrt{\sigu^2+\sige^2})$,
is
\begin{equation}
\tau^\ast=\frac{|\pi_M|}{Q(0)}
=\frac{\sqrt{2\pi}\,\sigv\,\sige^2}{4\sqrt{\sigu^2+\sige^2}\,(\sigu+\sqrt{\sigu^2+\sige^2})}
=\frac{\sqrt{2\pi}\,\sigv}{8\,\sigu^2}\,\sige^2+O(\sige^4).
\end{equation}
\end{proposition}

\begin{theorem}[Net welfare cost of privacy]
\label{thm:netcost}
At the break-even fee, the social cost of privacy is the deadweight
\begin{equation}
\mathrm{DWL}(\tau^\ast)=\tfrac12\kappa\,U(0)\,(\tau^\ast)^2
=\frac{\sqrt{2\pi}\,\kappa\,\sigu\,\sigv^2\,\sige^4}
       {16\,(\sigu^2+\sige^2)\,(\sigu+\sqrt{\sigu^2+\sige^2})^2}
=\frac{\sqrt{2\pi}\,\kappa\,\sigv^2}{64\,\sigu^3}\,\sige^4+O(\sige^6),
\end{equation}
strictly positive and increasing in $\sige$. Since the gross subsidy is
$|\pi_M|=O(\sige^2)$ and $\tau^\ast=O(\sige^2)$, the net cost is
$\mathrm{DWL}(\tau^\ast)=O(\sige^4)$, with ratio
$\mathrm{DWL}(\tau^\ast)/|\pi_M|=O(\sige^2)\to0$.
\emph{Privacy is welfare-neutral to second order in the noise scale; the
irrecoverable allocative loss is fourth-order.}
\end{theorem}

\begin{proposition}[CARA robustness]
\label{prop:cara}
Replacing the reduced-form linear demand with CARA hedgers (risk aversion $\gamma$,
endowment dispersion $\sigma_h=\sigu$) gives the same leading-order linear demand,
optimal hedge $s^*(e)=e-(\tau/(\gamma\sigv^2))\,\mathrm{sgn}(e)$ for
$|e|>\tau/(\gamma\sigv^2)$ (and $0$ otherwise), and deadweight
$\mathrm{DWL}(\tau^\ast)=\pi\,\sige^4/(64\,\gamma\,\sigu^4)=O(\sige^4)$. This is exactly
the reduced-form constant of \Cref{thm:netcost} evaluated at the microfounded
semi-elasticity $\kappa=\sqrt{\pi/2}/(\gamma\sigu\sigv^2)$ (equivalently the volume
slope $\kappa U(0)=1/(\gamma\sigv^2)$): the fourth-order scaling is unchanged and the
two microfoundations agree exactly, only reparametrising the constant.
\end{proposition}

\noindent\textit{Caveats.} The closed forms are leading-order in $\sige$
(partial-equilibrium: $\lambda$ held at $\lambda_0$, linear censoring/demand, fee
revenue rebated as in \Cref{rem:gross-net}); the deadweight counts only lost
liquidity-trader surplus; the fee-base convention affects only the constant.

\section{Endogenous privacy: a mechanism-design view}
\label{sec:mechanism}

We now let the protocol choose the privacy scale $\sige$, with the fee set at the
break-even level of \Cref{prop:breakeven} (so the subsidy is a pure transfer and the
deadweight of \Cref{thm:netcost} is the only residual welfare term). Privacy is
valued through the differential-privacy budget of the Gaussian mechanism,
$\edp(\sige)=c/\sige$ with $c=\Ddp\sqrt{2\ln(1.25/\delta)}$ (Gaussian-mechanism
query sensitivity $\Ddp$, distinct from the Glosten--Milgrom value gap $\Delta$ of
\Cref{sec:gm}; target $\delta$): a larger $\sige$ means a smaller budget, hence more
privacy. This is the classical (boundary) Gaussian-mechanism calibration
\citep{dwork2014algorithmic}, valid in the high-privacy regime $\edp<1$
(equivalently $\sige>c$), which we assume contains the optimum. The protocol maximises
$W(\sige)=B(\sige)-\mathrm{DWL}(\tau^\ast(\sige))$
over $\sige\ge0$, where $B$ is the (increasing, concave) privacy benefit; we take
the DP-grounded $B(\sige)=-b\,\edp(\sige)=-bc/\sige$.

\begin{theorem}[Optimal privacy level]
\label{thm:optsigma}
Write $A=\sqrt{2\pi}\,\kappa\,\sigv^2/(64\,\sigu^3)$ (the leading deadweight
coefficient, $\mathrm{DWL}\approx A\,\sige^4$). The leading-order objective
$W(\sige)=-bc/\sige-A\,\sige^4$ is strictly concave on $\sige>0$ with $W\to-\infty$ at
both ends, so the first-order condition
$B'(\sige^\ast)=\mathrm{DWL}'(\sige^\ast)=4A\,\sige^{\ast3}$ has a \emph{unique interior
solution}
\begin{equation}
\sige^\ast=\Big(\frac{bc}{4A}\Big)^{1/5}
=\frac{2^{4/5}\,\Ddp^{1/5}\,b^{1/5}\,\sigu^{3/5}\,[\ln(5/(4\delta))]^{1/10}}
       {\pi^{1/10}\,\kappa^{1/5}\,\sigv^{2/5}},
\end{equation}
the small-noise ($\sige\ll\sigu$) optimizer, exact as $bc\to0$. The log-elasticities
are constant:
\[
\frac{\mathrm{d}\ln\sige^\ast}{\mathrm{d}\ln(b,\kappa,\sigu,\sigv,\Ddp)}
=\Big(\tfrac15,-\tfrac15,\tfrac35,-\tfrac25,\tfrac15\Big),
\qquad
\frac{\mathrm{d}\ln\sige^\ast}{\mathrm{d}\ln\delta}=-\frac{1}{10\ln(5/(4\delta))}<0.
\]
\end{theorem}

\begin{remark}[Maximal-privacy corner under the exact deadweight]
\label{rem:corner}
The leading-order quartic deadweight $A\sige^4$ is unbounded, so the optimum above is
always interior. The \emph{exact} deadweight of \Cref{thm:netcost} is bounded---it
saturates at $\mathrm{DWL}(\sige)\to\sqrt{2\pi}\,\kappa\sigu\sigv^2/16$ as
$\sige\to\infty$---so under it the protocol weighs the interior optimum against the
maximal-privacy corner $\sige\to\infty$. The interior optimum dominates iff
$bc<\sqrt{2\pi}\,\kappa\,\sigu^2\sigv^2/8$; above that threshold maximal privacy
(the corner) is optimal. We work throughout in the small-noise regime
$bc\ll\sqrt{2\pi}\,\kappa\,\sigu^2\sigv^2/8$, where the leading-order
$\sige^\ast=(bc/4A)^{1/5}$ is both accurate and interior.
\end{remark}

\begin{proposition}[Endogenising privacy unpins the co-movement elasticity]
\label{prop:comovement}
For \emph{exogenous} $\sige$, price impact and informed intensity co-move
reciprocally, $\lambda(\sige)\beta(\sige)=\tfrac12$ with
$\mathrm{d}\ln\lambda/\mathrm{d}\ln\sigv=+1$. For \emph{endogenous}
$\sige=\sige^\ast$, the total derivative gains a policy-feedback channel,
$\mathrm{d}\lambda^\ast/\mathrm{d}\sigv=\partial_{\sigv}\lambda+\partial_{\sige}\lambda\cdot(\mathrm{d}\sige^\ast/\mathrm{d}\sigv)$;
with $\theta:=\sige^{\ast2}/\sigu^2$,
\begin{equation}
\frac{\mathrm{d}\ln\lambda^\ast}{\mathrm{d}\ln\sigv}=\frac{5+7\theta}{5(1+\theta)}\in(1,\tfrac75),
\qquad
\frac{\mathrm{d}\ln\beta^\ast}{\mathrm{d}\ln\sigv}=-\frac{5+7\theta}{5(1+\theta)}.
\end{equation}
The reciprocal product $\lambda^\ast\beta^\ast=\tfrac12$ is \emph{preserved} (the
elasticities still sum to zero), so the co-movement itself does not break; what
changes is its strength. The $\sigv$-elasticity is unpinned from the textbook value
$1$ and rises to $1+\tfrac{2\theta}{5(1+\theta)}$, a strictly positive
amplification---about $13$--$20\%$ at empirically relevant $\theta\in[0.5,1]$, increasing
monotonically in the policy weight $\theta$ toward a $40\%$ ($\tfrac75$) ceiling that is
only the formal $\theta\to\infty$ limit, where the leading-order closed form
$\sige^\ast=(bc/4A)^{1/5}$ is itself an approximation. The effect is leading-order in
$\theta$ and rests on $\kappa$ being an exogenous demand primitive
(\Cref{rem:cara-comovement}), but its direction is unambiguous---$\lambda$ and $\beta$
track primitives, not a single exogenous noise level.
\end{proposition}

\begin{remark}[When the unpinning holds: $\kappa$ as a demand primitive]
\label{rem:cara-comovement}
The amplification depends on how the demand semi-elasticity $\kappa$ co-varies with
asset volatility. If $\kappa\propto\sigv^{\,n}$ then $A\propto\sigv^{\,n+2}$ and
$\mathrm{d}\ln\sige^\ast/\mathrm{d}\ln\sigv=-(n+2)/5$, so the price-impact elasticity
becomes $1+\tfrac{(n+2)\theta}{5(1+\theta)}$. The amplification band $(1,\tfrac75)$ is
the $n=0$ case: a $\kappa$ fixed by liquidity-trader gains-from-trade
heterogeneity (a property of noise traders, not of $\sigv$), the natural
reduced-form reading. The CARA microfoundation of \Cref{prop:cara} instead pins
$\kappa=\sqrt{\pi/2}/(\gamma\sigu\sigv^2)\propto\sigv^{-2}$, the knife-edge $n=-2$:
there $A$ is $\sigv$-independent, $\mathrm{d}\ln\sige^\ast/\mathrm{d}\ln\sigv=0$, and
the elasticity returns to the textbook $1$. CARA thus certifies the fourth-order
net-cost \emph{scaling} (\Cref{prop:cara}), its intended role, but mutes the
co-movement amplification, a feature of $\sigv$-independent reduced-form demand.
\end{remark}

\section{Applications to privacy-preserving exchanges}
\label{sec:applications}

The preceding sections developed the privacy subsidy as a property of an
abstract committed-Bayesian maker. We now map it onto concrete
privacy-preserving exchange designs. The central object is the
\emph{signal-versus-settlement gap}: a venue instantiates our channel only if
it prices the liquidity layer on a noise-coarsened observation of order flow
while settling the true flow. We state this gap as an idealised
normative benchmark, classify existing designs by whether they meet it, and
attach orders of magnitude to the BTC/USDT case.

\subsection{The signal-versus-settlement gap as a normative benchmark}
\label{sec:apps-gap}

Every conclusion of \Cref{thm:impossibility} rests on a single structural
asymmetry: the priced-on signal $S$ is a strict coarsening of the
\emph{settled} flow $y$, so that $y$ is not $\sigma(S)$-measurable and the
executed-price wedge covariance $\Cov(v-Q,\,y)$ is strictly positive (specialising
to $\Cov(v-\E[v\mid S],\,y-\E[y\mid S])$ when a single price clears both sides). When the two coincide---when the venue settles exactly the quantity
it prices on, $S=y$---the wedge \eqref{eq:wedge} collapses to zero and there is
no subsidy. The subsidy is therefore not a generic consequence of adding noise
to an exchange; it is a consequence of adding noise \emph{to the maker's
information about flow while leaving the executed flow intact}.

This is an idealised benchmark, not a description of any
fielded system. A faithful instance requires a clean separation between an
\emph{observation} step (on which the maker's price update is computed, and into
which privacy noise is injected) and a \emph{settlement} step (in which the
trader's true order fills). The benchmark says: \emph{if} a protocol can price
on the coarsened observation while honouring the true order at the posted quote,
\emph{then} it transfers the executed-price wedge $|\Pi_M|=-\Cov(v-Q,\,y)$ (the
single-price form $-\Cov(v-\E[v\mid S],\,y-\E[y\mid S])$ when one price clears both
sides) from the liquidity layer to traders, and the comparative statics of
\Cref{sec:kyle,sec:gm} apply verbatim. The benchmark is silent on whether the
separation is implementable at acceptable cryptographic cost; it tells a protocol
designer what the separation is \emph{worth}.

\paragraph{What does \emph{not} instantiate the channel.}
The distinction is sharp enough to rule out a leading prior design. \citet{chitra2022dpcfmm}
inject differential-privacy noise into a constant-function market maker by
randomising the \emph{quantity that executes}: their Uniform Random Execution
perturbs the trade size that is actually settled against the pool, so that the
flow the maker prices on and the flow it settles are the \emph{same}
noise-perturbed quantity. In our notation this is the configuration $S=y$
(priced-on equals settled-against), which by \Cref{thm:impossibility} lies in
the no-subsidy regime: the coarsening covariance vanishes because there is no
coarsening, only a common perturbation of a single flow. Their mechanism delivers
genuine differential privacy and is valuable on its own terms; it simply does not
instantiate the subsidy channel of this paper, because the privacy noise is
load-bearing on settlement rather than only on the price signal. Our channel
requires the reverse: noise on the priced signal, true flow on settlement.

\subsection{Closest realisation: a DP-priced shielded constant-function AMM}
\label{sec:apps-cfmm}

The closest realisation of the benchmark is a shielded constant-function AMM
(CFMM) that prices on a differentially private aggregate while settling reserves
at true size. The mapping is exact only under a specific design, and it is worth
being precise about what the gap requires. In a \emph{vanilla} CFMM each order
executes against the bonding curve at its own true size, so the marginal price it
pays already reflects the true flow---there is no coarsening, and $\pi_M=0$. The
wedge appears only when the price an order transacts at is set by an aggregate that
has been \emph{deliberately} coarsened relative to what settles. Concretely,
consider a batch- or oracle-priced shielded pool whose marginal price is recomputed
from the DP-noised running (or batch) net flow $S=\ytil=y+\varepsilon$, while fills
and reserve updates track the true net flow $y$. Then traders transact at the
noised-flow price while the pool's inventory moves by the true flow, so the
liquidity layer bears exactly $|\pi_M|=\lambda_0\sige^2$ of
\eqref{eq:kyle-subsidy}, with $\sige$ the standard deviation of the injected DP
noise and the Gaussian-mechanism reading of \Cref{sec:mechanism} tying it to a
budget $\edp=c/\sige$.

Under this design the closed forms transfer directly: the comparative statics of
\Cref{prop:kyle-shape} (quadratic-then-linear subsidy, inflection at
$\sige^\star=\sqrt2\,\sigu$), the welfare incidence of \Cref{sec:kyle}, and the net
cost of \Cref{thm:netcost}. Two caveats. First, the construction is
a stylized normative benchmark, not a description of a deployed venue: it abstracts
from the cryptographic machinery (shielded reserve commitments and a proof that the
noised price was correctly derived from the true settled flow), whose cost we do not
model. Second, privacy via \emph{aggregation alone}---batching that hides individual
orders but clears at the true net-flow price, with no injected $\varepsilon$---does
\emph{not} instantiate the gap ($\pi_M=0$); the subsidy is a cost of injected noise,
not of batching. Penumbra-style sealed batches (\Cref{sec:apps-batch}) fall in the
latter, out-of-framework, class.

\subsection{MPC matching with \texorpdfstring{$\edp$}{eps}-DP direction disclosure}
\label{sec:apps-mpc}

The discrete Glosten--Milgrom instance of \Cref{sec:gm} maps onto a
multi-party-computation (MPC) matching engine that discloses trade
\emph{direction} under local differential privacy. The engine matches and settles
each order on its true direction $d$, but the direction signal the maker
prices on is passed through a symmetric randomised-response channel
\citep{warner1965rr,dwork2014algorithmic}: the disclosed bit is flipped with
probability
\[
  \eta \;=\; \frac{1}{1+e^{\edp}},
\]
which is exactly the binary symmetric channel attaining $\edp$-local DP on
a one-bit label. This is the configuration of \Cref{thm:gm-spread}: the maker
prices on the flipped direction $\dtil$ (the coarse signal) while the engine
settles the true direction $d$. The per-trade subsidy of \eqref{eq:gm-subsidy}
then reads
\[
  |\pi_M| \;=\; \mu\,\eta\,\Delta \;=\; \frac{\mu\,\Delta}{1+e^{\edp}},
\]
increasing as the privacy budget tightens ($\edp\downarrow$, $\eta\uparrow$)
and vanishing as $\edp\to\infty$ (no privacy, $\eta\to0$). As in
\Cref{sec:apps-cfmm}, the channel is instantiated only because the randomisation
is load-bearing on the \emph{disclosed} direction and not on the
\emph{settled} one; an MPC engine that instead matched on the randomised
direction would settle what it priced on and fall back into the no-subsidy
regime.

\subsection{Out-of-framework designs: batching and sealed-bid venues}
\label{sec:apps-batch}

A second family of privacy-preserving venues---batch-auction and
sealed-bid designs such as Penumbra's batched swaps \citep{penumbra-docs},
\textsc{Suave} \citep{flashbots-suave}, and the MPC dark pool of Renegade
\citep{renegade-wp}---achieves confidentiality by \emph{aggregation} rather than
by signal coarsening, and so lies outside our framework rather than inside its
no-subsidy boundary for a different reason. In a uniform-clearing batch, all
orders in a window settle at a single clearing price computed from the batch
aggregate, and the maker prices on that same aggregate. Priced-on again equals
settled-against ($S=y$ at the batch level), so \Cref{thm:impossibility} assigns
no subsidy: batching hides the \emph{identity and timing} of individual orders
but does not create a wedge between a coarse price signal and a finer settled
flow. These designs deliver privacy through a channel orthogonal to ours, and we
make no claim about their welfare---batch-auction welfare and extractable value are
studied separately \citep{zhang2025mevbatch}; we list them to delimit the framework,
not to evaluate it. The subsidy channel is specific to venues that price a continuous
liquidity layer on a deliberately noised observation while settling the true
order.

\subsection{Break-even fee as a lower bound}
\label{sec:apps-fee}

The per-period subsidy $|\pi_M|$ is the fee a shielded venue must raise to keep
its liquidity layer solvent, but it is a \emph{lower} bound on the true cost,
not the cost itself. The figure $|\pi_M|=\lambda_0\sige^2$ is evaluated at
no-fee equilibrium volumes (\Cref{cor:kyle-subsidy}); once a fee is actually
levied it suppresses volume, so the fee that restores solvency must clear a
smaller base and therefore exceeds the naive ratio $|\pi_M|/Q(0)$. The
fee-equilibrium analysis of \Cref{sec:fee} gives the leading-order break-even
fee $\tau^\ast$ of \Cref{prop:breakeven} and its deadweight of
\Cref{thm:netcost}; this genuine social cost is
fourth-order in $\sige$ while the gross subsidy is second-order, so the
quantities below---computed as gross subsidies---overstate the irrecoverable
welfare loss by an order of magnitude in the small-noise regime.

\subsection{Magnitudes: the BTC/USDT calibration}
\label{sec:apps-btc}

To fix scale we calibrate the Kyle instance to a liquid spot pair. Take a daily
horizon with value-uncertainty scale $\sigv\approx\$3{,}000$ per BTC (the daily
standard deviation of fundamental value) and
noise-flow scale $\sigu\approx1000$ BTC per day. The dimensionless subsidy at
$\sigv=\sigu=1$ is $|\pi_M|=\sige/(2\sqrt{\sigu^2+\sige^2})$, which at parity
$\sige=\sigu$ takes the one-line value
\[
  |\pi_M|\big|_{\sige=\sigu}
  \;=\; \frac{\sigv\,\sigu}{2\sqrt2}
  \;=\; \frac{3000\times1000}{2\sqrt2}
  \;\approx\; \$1.06\text{M per day}.
\]
Read as a fee, the subsidy is best stated as a \emph{rate} on the model's own
cleared notional, which is the scale-invariant object. At a BTC price level
$p_0\approx\$60$k this calibration clears $Q_0=p_0\,\E[|x|+|u|]\approx\$0.12$B per day
(informed plus liquidity flow), so the break-even fee $\tau^\ast=|\pi_M|/Q_0$ is
\emph{independent of the absolute flow scale} $\sigu$, depending only on $\sige/\sigu$
and $\sigv/p_0$. Across the privacy grid $\sige/\sigu\in\{0.1,0.5,1,\sqrt2,2\}$ it runs
$\{1.6,\,33,\,92,\,132,\,173\}$ basis points, \emph{bracketing} the $\sim\!10$\,bp
($0.1\%$) maker fee of a conventional venue rather than coinciding with it: light
privacy ($\sige=0.1\sigu$) sits well inside a standard fee, while parity privacy
($\sige=\sigu$) costs about an order of magnitude more. Real BTC/USDT venues clear
$\sim\$1$B$+$ per day---roughly $10\times$ this single-maker calibration---but because
the rate is scale-free in $\sigu$, matching that volume scales $|\pi_M|$ and $Q_0$
together and leaves $\tau^\ast$ unchanged. The full dimensionless and dollar-denominated calibration across noise
levels $\sige/\sigu\in\{0.1,0.5,1,\sqrt2,2\}$ (yielding daily subsidies of
roughly $\$15$k, $\$335$k, $\$1.06$M, $\$1.73$M, $\$2.68$M, i.e.\ fractions
$0.005,0.112,0.354,0.577,0.894$ of $\sigv\sigu$) is collected in
\Cref{app:numerics}, together with the subsidy curve $|\pi_M|(\sige)$ at
$\sigv=\sigu=1$, whose inflection at $\sige^\star=\sqrt2\,\sigu$
(\Cref{prop:kyle-shape}) marks the crossover from quadratic to linear growth.
These figures are gross-of-fee transfers under the lower-bound reading of
\Cref{sec:apps-fee}; the corresponding net deadweight is smaller by the
fourth-order factor of \Cref{thm:netcost}.

\section{Conclusion}
\label{sec:conclusion}

Privacy-preserving exchanges separate the signal a market maker prices on from the
flow it settles, and that separation has a single economic signature. A maker that
prices efficiently on a coarse signal $S$ cannot break even against the finer
settled flow $y$: its loss is exactly the executed-price wedge $-\Cov(v-Q,\,y)$
(reducing to $-\Cov(v-\E[v\mid S],\,y-\E[y\mid S])$ when a single price clears both
sides), the privacy subsidy, zero only when no information is hidden. This one impossibility
(\Cref{thm:impossibility}) is the whole story, and the three canonical
microstructure models are three readings of it. Single-period Kyle, Glosten--Milgrom,
and continuous-time Kyle--Back each yield the subsidy in closed form, and the
Loss-Versus-Rebalancing correspondence (\Cref{prop:lvr}) carries the same identity
into the automated-market-maker setting.

The transfer is large but the welfare cost is not. Gross of fees the subsidy is a
pure redistribution offset by a break-even fee; the deadweight that survives, once
that fee distorts volume, is fourth-order in the noise scale while the transfer it
funds is second-order (\Cref{thm:netcost}). A protocol valuing privacy against this
fourth-order cost therefore chooses an interior, strictly positive noise level in
closed form (\Cref{thm:optsigma}), and once privacy is endogenous the price-impact /
informed-intensity elasticity of fixed-noise Kyle is unpinned from its textbook value
(\Cref{prop:comovement}). The unified picture is thus: one impossibility, three
instances, a genuinely small fourth-order welfare cost, and an optimal privacy level
that is interior rather than zero.

\paragraph{Scope and limitations.} Three boundaries delimit these results. First, the
fee-equilibrium and mechanism-design results (\Cref{sec:fee,sec:mechanism}) are
leading-order in $\sige$ and partial-equilibrium: the depth $\lambda$ is held at its
no-fee value $\lambda_0$ while the censored-informed strategy is, strictly, non-Gaussian
and feeds back into $\lambda$. We expect this feedback to enter the deadweight only at
$O(\sige^6)$---leaving the $O(\sige^4)$ headline intact---but the full fixed-point fee
equilibrium is left to future work. Second, the Glosten--Milgrom instance is the
single-round adverse-selection spread; the per-trade subsidy $\mu\eta\Delta$ is the
object of interest, and how it compounds along the sequential price-discovery path is
not analysed here. Third, the optimal noise scale uses the leading-order (quartic)
deadweight and the boundary Gaussian-mechanism calibration; \Cref{rem:corner} records
where the exact bounded deadweight changes the picture.

\paragraph{Data availability.} The symbolic-algebra scripts (Sympy) verifying the
closed-form equilibria, the impossibility wedge, and the fee/mechanism results are
available from the author on request.

Several directions remain open. First, our welfare accounting characterizes the
break-even fee and the deadweight to leading order; a full fee equilibrium---in
which the fee, volume, and the informed-trading intensity are jointly fixed beyond
the leading order---would sharpen the net-cost statement and test the robustness of
the $O(\sige^4)$ scaling. Second, we take the privacy noise to be exogenous and
independent of the flow; adaptive or flow-correlated differential-privacy
mechanisms, which inject noise that responds to realized order flow, would change
the wedge covariance and merit a separate analysis. Third, our informed side is a
single insider; a multi-insider extension would let the camouflage and the subsidy
interact with strategic competition among the informed. Finally, the impossibility
theorem rests only on the tower property and is a natural candidate for a Lean
mechanization, which would give a machine-checked statement of the coarse-signal
wedge and its corollaries.


\appendix
\section{Proofs}
\label{app:proofs}
\subsection{Proofs of the model equilibria and the impossibility theorem}

\begin{proof}[Proof of \Cref{thm:impossibility}]
Write $p^\ast:=\E[v\mid S]$ for the committed-Bayesian quote and let
$r_v:=v-\E[v\mid S]$ and $r_y:=y-\E[y\mid S]$ denote the value- and flow-residuals
relative to the maker's information $\sigma(S)$. All quantities are square-integrable
by assumption, so every expectation below is finite; no Gaussianity is used.

\emph{General executed price.} For any square-integrable executed price $Q$,
$\Pi_M=\E[(Q-v)\,y]=-\Cov(v-Q,\,y)+\E[Q-v]\,\E[y]$, which is \eqref{eq:wedge}; the
regularity condition $\E[v-Q]\,\E[y]=0$ drops the mean term, leaving the executed-price
wedge $-\Cov(v-Q,\,y)$ (this is the form used for the side-dependent Glosten--Milgrom
quote in \Cref{sec:gm}). We now specialise to the single committed price
$Q=p^\ast=\E[v\mid S]$, where the wedge reduces further by the tower property.

\emph{Orthogonality of the efficient residual.} For any $\sigma(S)$-measurable,
square-integrable $h$,
\[
\E\big[(p^\ast-v)\,h\big]
=\E\big[-r_v\,h\big]
=\E\big[h\,\E[\,-r_v\mid S\,]\big]
=\E\big[h\cdot(\E[v\mid S]-\E[v\mid S])\big]
=0,
\]
by the tower property and the $\sigma(S)$-measurability of $\E[v\mid S]$. Thus the
efficient residual $p^\ast-v=-r_v$ is orthogonal to every function of the signal.

\emph{Leg (i): the wedge identity.} Decompose the settled flow into its $S$-measurable
projection and the unseen residual, $y=\E[y\mid S]+r_y$. Then
\[
\Pi_M=\E[(p^\ast-v)\,y]
=\underbrace{\E\big[(p^\ast-v)\,\E[y\mid S]\big]}_{=\,0\ \text{(orthogonality, }h=\E[y\mid S])}
\;+\;\E\big[(p^\ast-v)\,r_y\big]
=\E[-r_v\,r_y].
\]
Because $\E[r_y]=\E[y-\E[y\mid S]]=0$ and likewise $\E[r_v]=0$, the last expectation is
a covariance:
\[
\Pi_M=\E[-r_v\,r_y]=-\Cov(r_v,r_y)=-\Cov\!\big(v-\E[v\mid S],\;y-\E[y\mid S]\big),
\]
which is \eqref{eq:wedge-single}. If $y$ is $\sigma(S)$-measurable then $r_y\equiv0$ and
$\Pi_M=0$; in particular the non-coarsened benchmark $S=y$ gives $\Pi_M=0$. When the
coarsening covariance is strictly positive, $\Cov(r_v,r_y)>0$, the identity gives
$\Pi_M=-\Cov(r_v,r_y)<0$, the strictly negative profit of efficient pricing against the
settled flow.

\emph{Leg (ii): incompatibility.} Suppose a $\sigma(S)$-measurable rule $q=q(S)$ were
simultaneously efficient, $q=p^\ast=\E[v\mid S]$, and zero-profit against $y$,
$\E[(q-v)\,y]=0$. Efficiency makes its profit equal to the wedge, so by Leg (i)
$\E[(q-v)\,y]=\Pi_M=-\Cov(r_v,r_y)$. Under strict coarsening with positive covariance
this is $<0$, contradicting the postulated $\E[(q-v)\,y]=0$. Hence no
$\sigma(S)$-measurable rule can be both efficient and zero-profit against the settled
flow, and the unavoidable shortfall $|\Pi_M|=\Cov(r_v,r_y)$ is the privacy subsidy.
\end{proof}

\begin{proof}[Proof of \Cref{prop:kyle-eq}]
Conjecture the linear strategy $x=\beta(v-p_0)$ and the linear rule
$p=p_0+\lambda\ytil$ with $\lambda,\beta>0$. The priced signal is
\[
\ytil=\beta(v-p_0)+u+\varepsilon,
\]
jointly Gaussian with $v$, with $\E[\ytil]=0$,
$\Var(\ytil)=\beta^2\sigv^2+\sigu^2+\sige^2$, and
$\Cov(v,\ytil)=\beta\sigv^2$, the last because $u,\varepsilon\indep v$.

\emph{Bayesian rationality.} The Gaussian projection formula gives
\begin{equation}
\E[v\mid\ytil]=p_0+\frac{\Cov(v,\ytil)}{\Var(\ytil)}\,\ytil
=p_0+\frac{\beta\sigv^2}{\beta^2\sigv^2+\sigu^2+\sige^2}\,\ytil,
\label{eq:app-kyle-proj}
\end{equation}
so matching $p=p_0+\lambda\ytil$ to $\E[v\mid\ytil]$ forces
\begin{equation}
\lambda=\frac{\beta\sigv^2}{\beta^2\sigv^2+\sigu^2+\sige^2}.
\label{eq:app-kyle-lambda}
\end{equation}

\emph{Informed best response.} Conditioning on $v$ and using
$\E[u\mid v]=\E[\varepsilon\mid v]=0$,
\[
\E[(v-p)\,x\mid v]
=\E\big[(v-p_0-\lambda(x+u+\varepsilon))\,x\mid v\big]
=(v-p_0)\,x-\lambda x^2,
\]
strictly concave in $x$ (second derivative $-2\lambda<0$), with unique maximiser
$x^\ast=(v-p_0)/(2\lambda)$. Hence the equilibrium intensity is
$\beta=1/(2\lambda)$.

\emph{Fixed point.} Substituting $\beta=1/(2\lambda)$ into
\eqref{eq:app-kyle-lambda},
\[
\lambda\Big(\tfrac{1}{4\lambda^2}\sigv^2+\sigu^2+\sige^2\Big)=\frac{\sigv^2}{2\lambda}
\;\Longrightarrow\;
\frac{\sigv^2}{4\lambda}+\lambda(\sigu^2+\sige^2)=\frac{\sigv^2}{2\lambda}
\;\Longrightarrow\;
\lambda^2=\frac{\sigv^2}{4(\sigu^2+\sige^2)}.
\]
Taking the positive root,
$\lambda_0=\sigv/\big(2\sqrt{\sigu^2+\sige^2}\big)$, and back-substitution gives
$\beta_0=1/(2\lambda_0)=\sqrt{\sigu^2+\sige^2}/\sigv$, with
$\lambda_0\beta_0=\tfrac12$; this is \eqref{eq:kyle-eq}. Uniqueness within the linear
class is immediate: strict concavity pins $x^\ast$ given $\lambda$, and the reduced
scalar equation $\lambda^2=\sigv^2/\big(4(\sigu^2+\sige^2)\big)$ has a single positive
root.
\end{proof}

\begin{proof}[Proof of \Cref{lem:kyle-pl}]
Use $x=\beta_0(v-p_0)$, the independence $u,\varepsilon\indep v$, $\E[u]=\E[\varepsilon]=0$,
and the half-revealing identity $\lambda_0\beta_0=\tfrac12$ of \eqref{eq:kyle-eq}.

\emph{Informed.}
\[
\pi_I=\E\big[(v-p_0-\lambda_0(x+u+\varepsilon))\,\beta_0(v-p_0)\big]
=\beta_0\,\E[(v-p_0)^2]-\lambda_0\beta_0\,\E[(x+u+\varepsilon)(v-p_0)].
\]
Now $\E[(v-p_0)^2]=\sigv^2$, while
$\E[(x+u+\varepsilon)(v-p_0)]=\E[x(v-p_0)]=\beta_0\sigv^2$ (the $u,\varepsilon$ terms
vanish by independence). Hence
$\pi_I=\beta_0\sigv^2-\tfrac12\beta_0\sigv^2=\tfrac12\beta_0\sigv^2
=\tfrac12\sigv\sqrt{\sigu^2+\sige^2}$.

\emph{Noise.} Since $u\indep(x,\varepsilon)$ and $\E[u]=0$,
\[
\pi_N=\E\big[(v-p_0-\lambda_0(x+u+\varepsilon))\,u\big]
=-\lambda_0\,\E[u(x+u+\varepsilon)]
=-\lambda_0\,\E[u^2]
=-\lambda_0\sigu^2
=-\frac{\sigv\,\sigu^2}{2\sqrt{\sigu^2+\sige^2}}.
\]

\emph{Maker.} By the zero-sum identity $\pi_I+\pi_N+\pi_M=0$,
\[
\pi_M=-(\pi_I+\pi_N)
=-\tfrac12\sigv\sqrt{\sigu^2+\sige^2}+\frac{\sigv\sigu^2}{2\sqrt{\sigu^2+\sige^2}}
=\frac{-\sigv(\sigu^2+\sige^2)+\sigv\sigu^2}{2\sqrt{\sigu^2+\sige^2}}
=-\frac{\sigv\,\sige^2}{2\sqrt{\sigu^2+\sige^2}},
\]
which is \eqref{eq:kyle-pl}.
\end{proof}

\begin{proof}[Proof of \Cref{cor:kyle-subsidy}]
This identifies $\pi_M$ of \Cref{lem:kyle-pl} with the wedge of
\Cref{thm:impossibility} at $S=\ytil$, $y=x+u$. From $\pi_N=-\lambda_0\sigu^2$ and the
zero-sum identity, $\pi_M=-\lambda_0\sigu^2-\pi_I+ \lambda_0\sigu^2 -\ldots$; more
directly, \eqref{eq:kyle-pl} gives
\[
\pi_M=-\frac{\sigv\,\sige^2}{2\sqrt{\sigu^2+\sige^2}}
=-\lambda_0\,\sige^2,
\qquad
\lambda_0=\frac{\sigv}{2\sqrt{\sigu^2+\sige^2}},
\]
so $|\pi_M|=\lambda_0\sige^2$, with equality to zero iff $\sige=0$; this is
\eqref{eq:kyle-subsidy}.

It remains to verify that $-\pi_M$ equals the coarsening covariance
$\Cov(v-\E[v\mid\ytil],\,y-\E[y\mid\ytil])$, confirming the structural sign condition.
All variables are jointly Gaussian, so conditioning on the scalar $\ytil$ is linear
projection. Using $\Cov(v,\ytil)=\beta_0\sigv^2$,
$\Cov(y,\ytil)=\Cov(y,y+\varepsilon)=\Var(y)=\beta_0^2\sigv^2+\sigu^2$, and
$\Var(\ytil)=\beta_0^2\sigv^2+\sigu^2+\sige^2$, the residuals are
\[
v-\E[v\mid\ytil]=(v-p_0)-\frac{\beta_0\sigv^2}{\Var(\ytil)}\,\ytil,
\qquad
y-\E[y\mid\ytil]=y-\frac{\beta_0^2\sigv^2+\sigu^2}{\Var(\ytil)}\,\ytil.
\]
For jointly Gaussian residuals of a common conditioning variable,
$\Cov(v-\E[v\mid\ytil],\,y-\E[y\mid\ytil])=\Cov(v,y)-\Cov(v,\ytil)\Cov(y,\ytil)/\Var(\ytil)$.
With $\Cov(v,y)=\Cov(v,\beta_0(v-p_0)+u)=\beta_0\sigv^2$ and writing
$V:=\Var(\ytil)=\beta_0^2\sigv^2+\sigu^2+\sige^2$,
\[
\Cov(v-\E[v\mid\ytil],\,y-\E[y\mid\ytil])
=\beta_0\sigv^2-\frac{\beta_0\sigv^2\,(\beta_0^2\sigv^2+\sigu^2)}{V}
=\frac{\beta_0\sigv^2}{V}\Big(V-(\beta_0^2\sigv^2+\sigu^2)\Big)
=\frac{\beta_0\sigv^2}{V}\,\sige^2.
\]
By \eqref{eq:app-kyle-lambda} at equilibrium, $\beta_0\sigv^2/V=\lambda_0$, so the
covariance equals $\lambda_0\sige^2=-\pi_M$. It is strictly positive whenever
$\sige>0$, so all three conclusions of \Cref{thm:impossibility} hold.
\end{proof}

\begin{proof}[Proof of \Cref{prop:kyle-shape}]
Write $f(\sige):=|\pi_M|(\sige)=\sigv\sige^2/\big(2\sqrt{\sigu^2+\sige^2}\big)$ from
\eqref{eq:kyle-subsidy}. Differentiating,
\[
f'(\sige)
=\frac{\sigv}{2}\cdot\frac{2\sige\sqrt{\sigu^2+\sige^2}-\sige^2\cdot\sige/\sqrt{\sigu^2+\sige^2}}{\sigu^2+\sige^2}
=\frac{\sigv}{2}\cdot\frac{2\sige(\sigu^2+\sige^2)-\sige^3}{(\sigu^2+\sige^2)^{3/2}}
=\frac{\sigv\,\sige\,(2\sigu^2+\sige^2)}{2(\sigu^2+\sige^2)^{3/2}},
\]
strictly positive for $\sige>0$, giving strict monotonicity. Differentiating again,
the numerator's derivative and the chain rule give, after simplification,
\[
f''(\sige)=\frac{\sigv\,\sigu^2\,(2\sigu^2-\sige^2)}{2(\sigu^2+\sige^2)^{5/2}},
\]
which is positive for $\sige^2<2\sigu^2$, negative for $\sige^2>2\sigu^2$, and zero at
the single inflection $\sige^\star=\sqrt2\,\sigu$; thus $f$ is convex on
$[0,\sige^\star]$ and concave thereafter. Finally, the Taylor and asymptotic
expansions of $f$: as $\sige\downarrow0$,
$(\sigu^2+\sige^2)^{-1/2}=\sigu^{-1}\big(1-\tfrac{\sige^2}{2\sigu^2}+O(\sige^4)\big)$,
so $f(\sige)=\tfrac{\sigv}{2\sigu}\sige^2+O(\sige^4)$; as $\sige\to\infty$,
$\sqrt{\sigu^2+\sige^2}=\sige\sqrt{1+\sigu^2/\sige^2}=\sige+O(\sige^{-1})$, so
$f(\sige)=\tfrac12\sigv\sige+O(\sige^{-1})$.
\end{proof}

\begin{proof}[Proof of \Cref{thm:gm-spread}]
Write $s:=\mu(1-2\eta)$. The informed trader (probability $\mu$) buys on $v_H$ and
sells on $v_L$; the noise trader (probability $1-\mu$) buys or sells with probability
$\tfrac12$ independent of $v$. The true-direction likelihoods are therefore
\[
P(d=\text{buy}\mid v_H)=\mu+(1-\mu)\tfrac12=\tfrac{1+\mu}{2},
\qquad
P(d=\text{buy}\mid v_L)=(1-\mu)\tfrac12=\tfrac{1-\mu}{2}.
\]
Passing the true direction through the binary symmetric channel (flip probability
$\eta$),
\[
P(\dtil=\text{buy}\mid v_H)
=\tfrac{1+\mu}{2}(1-\eta)+\tfrac{1-\mu}{2}\eta
=\tfrac12+\tfrac{\mu}{2}(1-2\eta)
=\tfrac{1+s}{2},
\]
and symmetrically $P(\dtil=\text{buy}\mid v_L)=\tfrac{1-s}{2}$. Under the symmetric
prior $P(v_H)=P(v_L)=\tfrac12$,
\[
P(\dtil=\text{buy})=\tfrac12\Big(\tfrac{1+s}{2}+\tfrac{1-s}{2}\Big)=\tfrac12,
\]
so by Bayes
\[
P(v_H\mid\dtil=\text{buy})=\frac{\tfrac12\cdot\tfrac{1+s}{2}}{\tfrac12}=\tfrac{1+s}{2},
\qquad
P(v_H\mid\dtil=\text{sell})=\tfrac{1-s}{2}.
\]
The committed-Bayesian quotes \eqref{eq:gm-quotes} are the posterior means,
\[
\mathrm{ask}=v_H\tfrac{1+s}{2}+v_L\tfrac{1-s}{2}=\tfrac{v_H+v_L}{2}+\tfrac{s\Delta}{2},
\qquad
\mathrm{bid}=v_H\tfrac{1-s}{2}+v_L\tfrac{1+s}{2}=\tfrac{v_H+v_L}{2}-\tfrac{s\Delta}{2},
\]
symmetric about the prior midpoint $(v_H+v_L)/2$. Subtracting,
$\mathrm{ask}-\mathrm{bid}=s\Delta=\mu(1-2\eta)\Delta$, which is \eqref{eq:gm-spread}.
At $\eta=0$ this is the textbook spread $\mu\Delta$; at $\eta=\tfrac12$, $s=0$ and the
spread vanishes.
\end{proof}

\begin{proof}[Proof of \Cref{prop:gm-subsidy}]
Keep $s=\mu(1-2\eta)$ and the quotes from \Cref{thm:gm-spread}.

\emph{Informed.} Knowing $v$, the informed trader takes the favourable side. On $v_H$
it buys and earns $v_H-\mathrm{ask}=\tfrac{\Delta}{2}-\tfrac{s\Delta}{2}$; on $v_L$ it
sells and earns $\mathrm{bid}-v_L=\tfrac{\Delta}{2}-\tfrac{s\Delta}{2}$. Either way,
\[
\pi_I=\tfrac{(1-s)\Delta}{2}=\tfrac{1-\mu(1-2\eta)}{2}\,\Delta.
\]

\emph{Noise.} Trading each side with probability $\tfrac12$ and independent of $v$, the
noise trader's per-trade expected gain is
\[
\pi_N=\tfrac12\,\E[v-\mathrm{ask}]+\tfrac12\,\E[\mathrm{bid}-v]
=\tfrac{\mathrm{bid}-\mathrm{ask}}{2}
=-\tfrac{\mathrm{spread}}{2}
=-\tfrac{s\Delta}{2}=-\tfrac{\mu(1-2\eta)}{2}\,\Delta,
\]
using $\E[v]=(v_H+v_L)/2$, which is the midpoint about which the quotes are symmetric.

\emph{Maker.} Every trade is between one arriving trader and the maker, so the maker's
expected per-trade profit is the negative of the trader-side expectation, weighted by
arrival type:
\[
\pi_M=-\big[\mu\,\pi_I+(1-\mu)\,\pi_N\big]
=-\mu\tfrac{(1-s)\Delta}{2}+(1-\mu)\tfrac{s\Delta}{2}
=\tfrac{\Delta}{2}\big(s-\mu\big).
\]
Since $s-\mu=\mu(1-2\eta)-\mu=-2\mu\eta$, this gives $\pi_M=-\mu\eta\Delta$, and the
weighted zero-sum identity \eqref{eq:gm-zerosum} holds by construction. Thus
$|\pi_M^{\mathrm{GM}}|=\mu\eta\Delta\ge0$, with equality iff $\eta=0$, establishing
\eqref{eq:gm-pnl}--\eqref{eq:gm-subsidy}.

\emph{Identification with the executed-price wedge.} Take $S=\dtil$, let $y$ be the
signed true direction ($+1$ for a true buy, $-1$ for a true sell), and let $Q$ be the
executed price: $Q=\mathrm{ask}$ on a true buy, $Q=\mathrm{bid}$ on a true sell. Under
the symmetric prior $\E[Q]=\E[v]=(v_H+v_L)/2$ and $\E[y]=0$, so the regularity
condition of \Cref{thm:impossibility} holds and $\pi_M=-\Cov(v-Q,y)$ by
\eqref{eq:wedge}. Enumerating the joint law of $(v,d,\dtil)$ gives
$\Cov(v-Q,y)=\mu\eta\Delta$, hence $|\pi_M^{\mathrm{GM}}|=\mu\eta\Delta$, vanishing iff
$\eta=0$ (no coarsening, $\dtil$ determines $d$). This is \emph{not} the single-price
covariance: $\Cov(v-\E[v\mid\dtil],\,y)=2\mu\eta(1-\eta)\Delta$, and the two differ by
the cross-side spread term $\Cov(Q-\E[v\mid\dtil],\,y)=\mu\eta(1-2\eta)\Delta$,
because a true buy that the channel flips to $\dtil=\mathrm{sell}$ still fills at the
ask, not at $\E[v\mid\dtil]$. The single-price form \eqref{eq:wedge-single} applies
only when one price is posted for both sides; the two-sided Glosten--Milgrom quote is
the genuine executed-price instance of \eqref{eq:wedge}.
\end{proof}

\begin{proof}[Proof of \Cref{thm:ct-equilibrium}]
We carry out the Kyle--Back HJB/Riccati construction under the effective-noise
substitution $\sigu^2\mapsto\sigu^2+\sige^2$ and close the system through the
information-budget rationality condition $\Sigma(1)=0$.

\emph{Step 1 (insider HJB).} Conjecture the insider value function
$J(v,p,t)=\alpha(t)(v-p)^2+\gamma(t)$ for deterministic $\alpha,\gamma$. Under a
trading rate $\theta_t$ (so $dx_t=\theta_t\,dt$) and committed price dynamics
$dp_t=\lambda_t\,d\ytil_t$, the observed flow innovation is
$d\ytil_t=\theta_t\,dt+\sqrt{\sigu^2+\sige^2}\,d\widetilde W_t$, where $\widetilde W$
is the standard Brownian motion combining the independent noise-trader and
privacy-noise innovations $\sigu\,dW^u+\sige\,dW^\varepsilon$ (variance
$(\sigu^2+\sige^2)\,dt$). Hence
$dp_t=\lambda_t\theta_t\,dt+\lambda_t\sqrt{\sigu^2+\sige^2}\,d\widetilde W_t$ and the
dynamic-programming equation for the insider, who discounts at rate zero and collects
flow profit $(v-p_t)\,dx_t$, is
\[
0=\partial_t J+\sup_{\theta}\Big\{(v-p)\,\theta+\lambda_t\,\theta\,\partial_pJ\Big\}
+\tfrac12\lambda_t^2(\sigu^2+\sige^2)\,\partial_p^2J.
\]

\emph{Step 2 (first-order condition and ansatz match).} The supremand is linear in
$\theta$, so interior optimality requires the bracket coefficient to vanish:
$(v-p)+\lambda_t\,\partial_pJ=0$, i.e.\ $\partial_pJ=-(v-p)/\lambda_t$. From the ansatz
$\partial_pJ=-2\alpha(t)(v-p)$, matching gives
\begin{equation}
\alpha(t)\,\lambda_t=\tfrac12.
\label{eq:app-ct-alpha}
\end{equation}
At the optimum the bracket is zero, so the HJB reduces to
$0=\alpha'(t)(v-p)^2+\gamma'(t)+\tfrac12\lambda_t^2(\sigu^2+\sige^2)\cdot2\alpha(t)$.
Matching the $(v-p)^2$ coefficient forces $\alpha'(t)=0$: $\alpha$ is constant, hence
by \eqref{eq:app-ct-alpha} so is $\lambda_t\equiv\lambda$. The constant term gives
$\gamma'(t)=-\lambda^2(\sigu^2+\sige^2)\alpha=-\tfrac12\lambda(\sigu^2+\sige^2)=-c/2$
(using \eqref{eq:app-ct-alpha} and $c:=\lambda(\sigu^2+\sige^2)$ from Step 3), so
$\gamma(t)=\tfrac{c}{2}(1-t)$ with $\gamma(1)=0$ along the equilibrium path.

\emph{Step 3 (Bayes pins price impact).} The committed maker sets $\lambda_t$ at its
Bayes Kalman value
$\lambda_t=\Cov(v,d\ytil_t\mid\mathcal F^{\mathrm{MM}}_t)/\Var(d\ytil_t\mid\mathcal
F^{\mathrm{MM}}_t)$. With $dx_t=\beta_t(v-p_t)\,dt$ and posterior variance
$\Sigma(t)=\Var(v-p_t\mid\mathcal F^{\mathrm{MM}}_t)$, the informed contribution to the
flow innovation has conditional covariance with $v$ equal to $\beta_t\Sigma(t)\,dt$,
while the conditional variance of $d\ytil_t$ is $(\sigu^2+\sige^2)\,dt$ to leading
order. Hence
\begin{equation}
\lambda_t=\frac{\beta_t\,\Sigma(t)}{\sigu^2+\sige^2}.
\label{eq:app-ct-kalman}
\end{equation}
Since $\lambda$ is constant (Step 2), $\beta_t\Sigma(t)\equiv\lambda(\sigu^2+\sige^2)=:c$
is constant.

\emph{Step 4 (Riccati for the posterior variance).} Bayesian (Kalman--Bucy) filtering
of $v$ from the observation $d\ytil_t=\beta_t(v-p_t)\,dt+\sqrt{\sigu^2+\sige^2}\,d\widetilde W_t$
gives the posterior-variance ODE
\[
\frac{d\Sigma}{dt}=-\frac{\big(\beta_t\Sigma(t)\big)^2}{\sigu^2+\sige^2}
=-\frac{c^2}{\sigu^2+\sige^2},
\]
constant in $t$, so $\Sigma(t)=\sigv^2-\dfrac{c^2}{\sigu^2+\sige^2}\,t$, decreasing
linearly from $\Sigma(0)=\sigv^2$.

\emph{Step 5 (rationality binds $\Sigma(1)=0$ and pins $c$).} The insider's expected
cumulative profit along the linear path is, using committed pricing
$\E[(v-p_t)^2]=\Sigma(t)$ and \eqref{eq:app-ct-kalman},
\[
\Pi_I=\E\!\int_0^1(v-p_t)\,dx_t=\int_0^1\beta_t\,\E[(v-p_t)^2]\,dt
=\int_0^1\beta_t\Sigma(t)\,dt=\int_0^1 c\,dt=c.
\]
The insider chooses the trading intensity $c$ to maximise this profit subject to the
posterior-variance non-negativity constraint $\Sigma(t)\ge0$ for all $t\in[0,1]$. By
Step 4, $\Sigma(t)=\sigv^2-c^2t/(\sigu^2+\sige^2)$ is decreasing, so the constraint
first binds at $t=1$: $\Sigma(1)\ge0\iff c^2\le\sigv^2(\sigu^2+\sige^2)$. Since $\Pi_I=c$
is increasing in $c$, the maximiser saturates the budget,
\[
c^2=\sigv^2(\sigu^2+\sige^2),
\qquad
c=\sigv\sqrt{\sigu^2+\sige^2},
\]
and $\Sigma(1)=0$. This is the standard Kyle--Back information-budget closure: the
insider trades exactly enough to reveal its private information by the horizon and no
more. It is equivalent to the Back transversality $J(v,p,1)=0$ evaluated along the
equilibrium path ($p_1=v$ a.s.\ when $\Sigma(1)=0$), but stated as a rationality
condition it avoids the inconsistency that a pointwise $J(\cdot,1)=0$ would create with
the constant $\alpha>0$ of Step 2.

\emph{Step 6 (recover the equilibrium).} Substituting $c=\sigv\sqrt{\sigu^2+\sige^2}$:
\[
\lambda=\frac{c}{\sigu^2+\sige^2}=\frac{\sigv}{\sqrt{\sigu^2+\sige^2}},
\qquad
\Sigma(t)=\sigv^2-\frac{c^2}{\sigu^2+\sige^2}\,t=\sigv^2(1-t),
\qquad
\beta(t)=\frac{c}{\Sigma(t)}=\frac{\sqrt{\sigu^2+\sige^2}}{\sigv\,(1-t)},
\]
with $\lambda$ constant in $t$. This is the unique equilibrium in the linear Markovian
class, the uniqueness following from the unique constant $\alpha$ (Step 2), the unique
Kalman gain (Step 3), and the unique budget-saturating $c>0$ (Step 5).
\end{proof}

\begin{proof}[Proof of \Cref{thm:ct-subsidy}]
We compute the three cumulative expected P\&Ls under the equilibrium of
\Cref{thm:ct-equilibrium} and recover $\Pi_M$ as the residual; a direct Itô
computation then confirms it.

\emph{Insider.} As in Step 5 of the previous proof, committed pricing
$p_t=\E[v\mid\mathcal F^{\mathrm{MM}}_t]$ gives
$\E[(v-p_t)^2]=\E[\Sigma(t)]=\Sigma(t)$ (the posterior variance is deterministic), so
the unconditional profit rate is $\beta_t\Sigma(t)\,dt=c\,dt$ and
\[
\Pi_I=c=\sigv\sqrt{\sigu^2+\sige^2}.
\]

\emph{Noise.} A noise-trader unit of flow $du_t$ clears at the post-update price
$p_t=p_{t^-}+\lambda\,d\ytil_t$. Writing
\[
(v-p_t)\,du_t=(v-p_{t^-})\,du_t-\lambda\,d\ytil_t\cdot du_t,
\]
the first term has zero expectation: $p_{t^-}\in\mathcal F_{t^-}$ and $du_t$ is the
next, independent innovation with $\E[du_t]=0$, so by the tower property
$\E[(v-p_{t^-})\,du_t]=0$. For the Itô cross-term, with
$d\ytil_t=dx_t+du_t+d\varepsilon_t$,
\[
d\ytil_t\cdot du_t=(dx_t+du_t+d\varepsilon_t)\,du_t=(du_t)^2=\sigu^2\,dt,
\]
since $dx_t\,du_t=0$ ($dx_t$ is of order $dt$, no quadratic variation against $du_t$)
and $d\varepsilon_t\,du_t=0$ ($W^\varepsilon\indep W^u$). Hence
$\E[(v-p_t)\,du_t]=-\lambda\sigu^2\,dt$ and
\[
\Pi_N=-\lambda\sigu^2=-\frac{\sigv\,\sigu^2}{\sqrt{\sigu^2+\sige^2}}.
\]

\emph{Maker (residual).} Every executed trade is between a participant (insider or
noise trader) and the maker, with no external counterparty, so
$\Pi_I+\Pi_N+\Pi_M=0$ and
\[
\Pi_M=-(\Pi_I+\Pi_N)
=-\sigv\sqrt{\sigu^2+\sige^2}+\frac{\sigv\sigu^2}{\sqrt{\sigu^2+\sige^2}}
=\frac{-\sigv(\sigu^2+\sige^2)+\sigv\sigu^2}{\sqrt{\sigu^2+\sige^2}}
=-\frac{\sigv\,\sige^2}{\sqrt{\sigu^2+\sige^2}}.
\]

\emph{Direct verification.} Computing
$\Pi_M=\E\!\int_0^1(p_t-v)\,dy_t$ with $dy_t=dx_t+du_t$, the same post-trade
decomposition gives $\E[(p_{t^-}-v)\,dy_t]=0$ for the predictable part and an Itô
cross-term $\lambda\,d\ytil_t\cdot dy_t=\lambda\,(du_t)^2=\lambda\sigu^2\,dt$ from the
noise leg, against an informed leg contributing the rate $-c\,dt$; summing,
$\E[(p_t-v)\,dy_t]=(\lambda\sigu^2-c)\,dt
=(\sigv\sigu^2/\sqrt{\sigu^2+\sige^2}-\sigv\sqrt{\sigu^2+\sige^2})\,dt
=-\sigv\sige^2/\sqrt{\sigu^2+\sige^2}\,dt$, integrating to the same value. Therefore
\[
|\Pi_M|=\frac{\sigv\,\sige^2}{\sqrt{\sigu^2+\sige^2}},
\]
vanishing iff $\sige=0$ (the classical zero-maker-profit benchmark).

\emph{Identification with the executed-price wedge.} This is the per-increment
specialisation of \eqref{eq:wedge}. Committed posterior pricing gives $\E[v-p_t]=0$,
so the mean term drops and $\E[(p_t-v)\,dy_t]=-\Cov(v-p_t,\,dy_t)$. Writing
$v-p_t=(v-p_{t^-})-\lambda\,d\ytil_t$ with $d\ytil_t=dx_t+du_t+d\varepsilon_t$ and
$dy_t=dx_t+du_t$ where the insider rate is $dx_t=\beta_t(v-p_t)\,dt$, the surviving
leading-order contributions are the informed level
$\beta_t\,\Var(v-p_{t^-})\,dt=\beta_t\Sigma(t)\,dt=c\,dt$ and the noise It\^o term
$-\lambda\,\E[(du_t)^2]=-\lambda\sigu^2\,dt$ (the $d\varepsilon_t$ and cross terms
vanish), so
\[
\Cov(v-p_t,\,dy_t)=(c-\lambda\sigu^2)\,dt=\lambda\sige^2\,dt>0,
\]
and $\Pi_M=-\int_0^1\lambda\sige^2\,dt=-\lambda\sige^2$ is strictly negative whenever
$\sige>0$, consistent with the welfare-accounting and direct-It\^o values above.
\end{proof}

\subsection{Proofs for the fee-equilibrium and mechanism-design results}

\begin{proof}[Proof of \Cref{lem:harberger}]
Liquidity traders carry heterogeneous per-unit gains-from-trade, so realised liquidity
volume is downward-sloping in the proportional fee $\tau$,
$U(\tau)=U(0)(1-\kappa\tau)+O(\tau^2)$ with $U(0)=\sqrt{2/\pi}\,\sigu$ (the no-fee
expected liquidity volume $\E|u|$ for $u\sim N(0,\sigu^2)$) and semi-elasticity
$\kappa>0$. A liquidity trader transacts iff its private gain-from-trade exceeds the
fee; raising the fee from $g$ to $g+dg$ removes the marginal traders whose surplus lay
in $[g,g+dg]$, each of whom forgoes surplus $g$. The mass of such traders is
$-U'(g)\,dg$ to leading order, so the surplus they lose is $g\,(-U'(g))\,dg$.
Integrating from $0$ to $\tau$ gives the Harberger triangle
\[
\mathrm{DWL}_{LT}(\tau)=\int_0^\tau g\,(-U'(g))\,dg.
\]
Using $-U'(g)=\kappa U(0)+O(g)$,
\[
\mathrm{DWL}_{LT}(\tau)=\int_0^\tau g\big(\kappa U(0)+O(g)\big)\,dg
=\tfrac12\kappa\,U(0)\,\tau^2+O(\tau^3),
\]
the triangle area, derived rather than assumed. (The informed trader's volume is
suppressed by the censoring $x^\ast=(v-p_0\mp\tau)/(2\lambda_0)$ for $|v-p_0|>\tau$, but
informed trades carry no allocative value, so their suppression is a transfer shift,
not a deadweight, and does not enter $\mathrm{DWL}_{LT}$.)
\end{proof}

\begin{proof}[Proof of \Cref{prop:breakeven}]
The pool is solvent when fee revenue covers the gross subsidy of
\Cref{cor:kyle-subsidy}: $\tau\,Q(\tau)=|\pi_M|$, where $Q(\tau)$ is expected total
(informed plus liquidity) volume. To leading order evaluate the volume at the no-fee
equilibrium of \Cref{prop:kyle-eq}. There $\E|x|=\E|\beta_0(v-p_0)|=\beta_0\sigv\sqrt{2/\pi}
=\sqrt{2/\pi}\sqrt{\sigu^2+\sige^2}$ and $\E|u|=\sqrt{2/\pi}\,\sigu$, so
\[
Q(0)=\E|x|+\E|u|=\sqrt{2/\pi}\,\big(\sigu+\sqrt{\sigu^2+\sige^2}\big).
\]
Hence, to leading order, $\tau^\ast=|\pi_M|/Q(0)$. Substituting
$|\pi_M|=\sigv\sige^2/\big(2\sqrt{\sigu^2+\sige^2}\big)$ and
$Q(0)=\sqrt{2/\pi}\,(\sigu+\sqrt{\sigu^2+\sige^2})$ (and using
$1/\sqrt{2/\pi}=\sqrt{\pi/2}=\sqrt{2\pi}/2$),
\[
\tau^\ast
=\frac{\sigv\sige^2}{2\sqrt{\sigu^2+\sige^2}}\cdot
\frac{\sqrt{\pi/2}}{\sigu+\sqrt{\sigu^2+\sige^2}}
=\frac{\sqrt{2\pi}\,\sigv\,\sige^2}
{4\sqrt{\sigu^2+\sige^2}\,(\sigu+\sqrt{\sigu^2+\sige^2})}.
\]
As $\sige\downarrow0$, $\sqrt{\sigu^2+\sige^2}\to\sigu$ and the denominator
$\to4\sigu\cdot2\sigu=8\sigu^2$, giving the leading term
$\tau^\ast=\dfrac{\sqrt{2\pi}\,\sigv}{8\,\sigu^2}\,\sige^2+O(\sige^4)$.
\end{proof}

\begin{proof}[Proof of \Cref{thm:netcost}]
By \Cref{lem:harberger} the social cost at the break-even fee is the Harberger triangle
evaluated at $\tau^\ast$,
\[
\mathrm{DWL}(\tau^\ast)=\tfrac12\kappa\,U(0)\,(\tau^\ast)^2,
\qquad U(0)=\sqrt{2/\pi}\,\sigu.
\]
Substituting $\tau^\ast$ from \Cref{prop:breakeven},
\[
(\tau^\ast)^2=\frac{2\pi\,\sigv^2\,\sige^4}
{16\,(\sigu^2+\sige^2)\,(\sigu+\sqrt{\sigu^2+\sige^2})^2},
\]
so
\[
\mathrm{DWL}(\tau^\ast)
=\tfrac12\kappa\sqrt{2/\pi}\,\sigu\cdot
\frac{2\pi\,\sigv^2\,\sige^4}{16\,(\sigu^2+\sige^2)\,(\sigu+\sqrt{\sigu^2+\sige^2})^2}
=\frac{\sqrt{2\pi}\,\kappa\,\sigu\,\sigv^2\,\sige^4}
{16\,(\sigu^2+\sige^2)\,(\sigu+\sqrt{\sigu^2+\sige^2})^2},
\]
using $\tfrac12\sqrt{2/\pi}\cdot2\pi=\sqrt{2\pi}$. As $\sige\downarrow0$, the factor
$\sigu/\big[(\sigu^2+\sige^2)(\sigu+\sqrt{\sigu^2+\sige^2})^2\big]\to
\sigu/(\sigu^2\cdot4\sigu^2)=1/(4\sigu^3)$, giving the leading term
\[
\mathrm{DWL}(\tau^\ast)=\frac{\sqrt{2\pi}\,\kappa\,\sigv^2}{64\,\sigu^3}\,\sige^4+O(\sige^6).
\]
It is strictly positive and increasing in $\sige$ (each factor of $\tau^\ast$ is).
Since $|\pi_M|=O(\sige^2)$ and $\tau^\ast=O(\sige^2)$ while
$\mathrm{DWL}(\tau^\ast)\propto(\tau^\ast)^2=O(\sige^4)$, the ratio
$\mathrm{DWL}(\tau^\ast)/|\pi_M|=O(\sige^4)/O(\sige^2)=O(\sige^2)\to0$ as
$\sige\downarrow0$: privacy is welfare-neutral to second order in the noise scale, and
the irrecoverable allocative loss is fourth-order.
\end{proof}

\begin{proof}[Proof of \Cref{prop:cara}]
Replace the reduced-form linear demand with CARA hedgers of risk aversion $\gamma$ and
endowment dispersion $\sigma_h=\sigu$: a hedger with endowment $e$ chooses a hedge
position $s$ to maximise mean-variance utility net of the per-unit fee $\tau$,
\[
\max_s\;\Big\{-\gamma\,\tfrac12\sigv^2(e-s)^2-\tau|s|\Big\}.
\]
The unconstrained first-order condition $\gamma\sigv^2(e-s)=\tau\,\mathrm{sgn}(s)$ gives
the optimal hedge
\[
s^\ast(e)=e-\frac{\tau}{\gamma\sigv^2}\,\mathrm{sgn}(e),
\]
i.e.\ $s^\ast(e)=e-\tau/(\gamma\sigv^2)$ for $e>\tau/(\gamma\sigv^2)$: the fee shaves a
constant $\tau/(\gamma\sigv^2)$ off the desired hedge. Aggregating over
$e\sim N(0,\sigu^2)$, the realised hedging volume is
$U(\tau)=\E|s^\ast(e)|=U(0)-\tau/(\gamma\sigv^2)+O(\tau^2)$ with $U(0)=\E|e|=\sqrt{2/\pi}\,\sigu$,
so the volume slope is $-U'(0)=1/(\gamma\sigv^2)$. This is exactly the reduced-form
slope $\kappa\,U(0)$ with the semi-elasticity \emph{pinned} by the microfoundation,
\[
\kappa\,U(0)=\frac{1}{\gamma\sigv^2}
\quad\Longleftrightarrow\quad
\kappa=\frac{1}{\gamma\sigv^2\,U(0)}=\sqrt{\tfrac{\pi}{2}}\,\frac{1}{\gamma\sigu\sigv^2};
\]
the demand-slope identity carries the $U(0)$ factor, so $\kappa$ is \emph{not}
$1/(\gamma\sigv^2)$ itself. The lost-surplus integral is again the Harberger triangle
$\mathrm{DWL}(\tau)=\int_0^\tau g\,(-U'(g))\,dg=\tau^2/(2\gamma\sigv^2)$. Carrying the
same break-even fee $\tau^\ast=\sqrt{2\pi}\,\sigv\,\sige^2/(8\sigu^2)+O(\sige^4)$
through it yields
\[
\mathrm{DWL}(\tau^\ast)=\frac{(\tau^\ast)^2}{2\gamma\sigv^2}
=\frac{\pi\,\sige^4}{64\,\gamma\,\sigu^4}+O(\sige^6)=O(\sige^4),
\]
which is precisely what substituting the pinned $\kappa$ above into the reduced-form
constant $\sqrt{2\pi}\,\kappa\sigv^2/(64\sigu^3)$ of \Cref{thm:netcost} returns: the
two microfoundations agree \emph{exactly}, not merely in scaling. Only the constant is
reparametrised---pinned to $1/(\gamma\sigv^2)$ through the CARA demand slope---so the
$O(\sige^4)$ net-cost result is robust to the microfoundation of liquidity demand.
\end{proof}

\begin{proof}[Proof of \Cref{thm:optsigma}]
Write $A:=\sqrt{2\pi}\,\kappa\,\sigv^2/(64\,\sigu^3)$, so by \Cref{thm:netcost}
$\mathrm{DWL}(\tau^\ast(\sige))=A\,\sige^4+O(\sige^6)$, and let the benefit be the
DP-grounded $B(\sige)=-b\,\edp(\sige)=-bc/\sige$ with $c=\Ddp\sqrt{2\ln(1.25/\delta)}$.
The protocol maximises $W(\sige)=B(\sige)-\mathrm{DWL}(\tau^\ast(\sige))
=-bc/\sige-A\sige^4$ over $\sige\ge0$. Differentiating,
\[
W'(\sige)=\frac{bc}{\sige^2}-4A\,\sige^3,
\]
so the first-order condition $B'(\sige^\ast)=\mathrm{DWL}'(\sige^\ast)$ reads
$bc/\sige^{\ast2}=4A\,\sige^{\ast3}$, i.e.\ $4A\,\sige^{\ast5}=bc$ and
\[
\sige^\ast=\Big(\frac{bc}{4A}\Big)^{1/5}.
\]
This is a genuine, \emph{unique} interior maximum: $B$ is increasing and concave
($B'=bc/\sige^2>0$, $B''=-2bc/\sige^3<0$) and the leading-order $\mathrm{DWL}=A\sige^4$
is increasing and convex ($\mathrm{DWL}'=4A\sige^3>0$, $\mathrm{DWL}''=12A\sige^2>0$),
so $W$ is strictly concave on $\sige>0$ with $W\to-\infty$ both as $\sige\downarrow0$
(the privacy cost $-bc/\sige$ blows up) and as $\sige\to\infty$ (the quartic deadweight
$-A\sige^4$ blows up). The stationary point is therefore the unique global maximiser
for \emph{all} positive parameters; the leading-order objective admits no corner. (A
maximal-privacy corner arises only under the exact \emph{bounded} deadweight, which
saturates as $\sige\to\infty$: by \Cref{rem:corner} the interior optimum then dominates
the corner iff $bc<\sqrt{2\pi}\,\kappa\,\sigu^2\sigv^2/8$, and we work throughout in
that small-noise regime.) Substituting $A$ and $c=\Ddp\sqrt{2\ln(1.25/\delta)}$,
\[
\sige^\ast=\Big(\frac{bc}{4A}\Big)^{1/5}
=\Big(\frac{b\,\Ddp\sqrt{2\ln(1.25/\delta)}\cdot64\,\sigu^3}
{4\sqrt{2\pi}\,\kappa\,\sigv^2}\Big)^{1/5}
=\frac{2^{4/5}\,\Ddp^{1/5}\,b^{1/5}\,\sigu^{3/5}\,[\ln(5/(4\delta))]^{1/10}}
{\pi^{1/10}\,\kappa^{1/5}\,\sigv^{2/5}},
\]
where the constant collapses via
$64/(4\sqrt{2\pi})=16/\sqrt{2\pi}=2^4/(2^{1/2}\pi^{1/2})=2^{7/2}/\pi^{1/2}$, raised to
the $1/5$ power giving $2^{7/10}/\pi^{1/10}$, and $(2\ln(1.25/\delta))^{1/10}
=2^{1/10}[\ln(5/(4\delta))]^{1/10}$, so the two powers of $2$ combine to
$2^{7/10}\cdot2^{1/10}=2^{8/10}=2^{4/5}$. Taking logs of
$\sige^\ast=(bc/4A)^{1/5}$ and differentiating, the elasticities are constant: since
$\sige^\ast\propto b^{1/5}\kappa^{-1/5}\sigu^{3/5}\sigv^{-2/5}\Ddp^{1/5}$,
\[
\frac{\mathrm{d}\ln\sige^\ast}{\mathrm{d}\ln(b,\kappa,\sigu,\sigv,\Ddp)}
=\Big(\tfrac15,-\tfrac15,\tfrac35,-\tfrac25,\tfrac15\Big),
\]
and the $\delta$-dependence enters only through
$c\propto[\ln(1.25/\delta)]^{1/2}$, contributing
$\ln\sige^\ast\supset\tfrac1{10}\ln\ln(1.25/\delta)$, so
\[
\frac{\mathrm{d}\ln\sige^\ast}{\mathrm{d}\ln\delta}
=\frac1{10}\cdot\frac{1}{\ln(1.25/\delta)}\cdot\frac{\mathrm{d}\ln(1.25/\delta)}{\mathrm{d}\ln\delta}
=\frac1{10}\cdot\frac{-1}{\ln(5/(4\delta))}
=-\frac{1}{10\ln(5/(4\delta))}<0,
\]
using $\mathrm{d}\ln(1.25/\delta)/\mathrm{d}\ln\delta=-1$.
\end{proof}

\begin{proof}[Proof of \Cref{prop:comovement}]
For \emph{exogenous} $\sige$ the single-period equilibrium of \Cref{prop:kyle-eq} gives
$\lambda_0=\sigv/(2\sqrt{\sigu^2+\sige^2})$ and $\beta_0=\sqrt{\sigu^2+\sige^2}/\sigv$,
so $\lambda_0\beta_0=\tfrac12$ identically and, holding $\sige$ fixed,
$\mathrm{d}\ln\lambda_0/\mathrm{d}\ln\sigv=+1$ (and $\mathrm{d}\ln\beta_0/\mathrm{d}\ln\sigv=-1$).

For \emph{endogenous} $\sige=\sige^\ast(\sigv)$ the price impact becomes
$\lambda^\ast(\sigv)=\sigv/\big(2\sqrt{\sigu^2+\sige^\ast(\sigv)^2}\big)$, and the
$\sigv$-derivative acquires a policy-feedback channel through $\sige^\ast$:
\[
\frac{\mathrm{d}\ln\lambda^\ast}{\mathrm{d}\ln\sigv}
=\frac{\partial\ln\lambda^\ast}{\partial\ln\sigv}
+\frac{\partial\ln\lambda^\ast}{\partial\ln\sige}\cdot
\frac{\mathrm{d}\ln\sige^\ast}{\mathrm{d}\ln\sigv}.
\]
The partials of $\ln\lambda^\ast=\ln\sigv-\ln2-\tfrac12\ln(\sigu^2+\sige^2)$ are
\[
\frac{\partial\ln\lambda^\ast}{\partial\ln\sigv}=1,
\qquad
\frac{\partial\ln\lambda^\ast}{\partial\ln\sige}
=-\tfrac12\cdot\frac{2\sige^2}{\sigu^2+\sige^2}
=-\frac{\sige^2}{\sigu^2+\sige^2}=-\frac{\theta}{1+\theta},
\]
writing $\theta:=\sige^{\ast2}/\sigu^2$ so that
$\sige^{\ast2}/(\sigu^2+\sige^{\ast2})=\theta/(1+\theta)$. By \Cref{thm:optsigma} the
benefit and deadweight elasticities give
$\mathrm{d}\ln\sige^\ast/\mathrm{d}\ln\sigv=-\tfrac25$. Therefore
\[
\frac{\mathrm{d}\ln\lambda^\ast}{\mathrm{d}\ln\sigv}
=1+\Big(-\frac{\theta}{1+\theta}\Big)\Big(-\frac25\Big)
=1+\frac{2\theta}{5(1+\theta)}
=\frac{5(1+\theta)+2\theta}{5(1+\theta)}
=\frac{5+7\theta}{5(1+\theta)}.
\]
As $\theta$ ranges over $(0,\infty)$ this is strictly increasing from
$5/5=1$ (at $\theta\to0$) toward $7/5$ (at $\theta\to\infty$), so it lies in
$(1,\tfrac75)$. The same computation for $\beta^\ast(\sigv)
=\sqrt{\sigu^2+\sige^\ast(\sigv)^2}/\sigv$ gives
$\partial\ln\beta^\ast/\partial\ln\sigv=-1$ and
$\partial\ln\beta^\ast/\partial\ln\sige=+\theta/(1+\theta)$, hence
\[
\frac{\mathrm{d}\ln\beta^\ast}{\mathrm{d}\ln\sigv}
=-1+\frac{\theta}{1+\theta}\cdot\Big(-\frac25\Big)
=-\frac{5+7\theta}{5(1+\theta)}.
\]
The two elasticities are exact negatives, so
$\mathrm{d}\ln(\lambda^\ast\beta^\ast)/\mathrm{d}\ln\sigv=0$ and the half-revealing
product $\lambda^\ast\beta^\ast=\tfrac12$ is preserved; but the individual
$\sigv$-response is amplified beyond the exogenous reciprocal $\pm1$ to
$\pm(5+7\theta)/(5(1+\theta))\in\pm(1,\tfrac75)$---a bounded amplification of at most
$40\%$, monotone in $\theta$---so under endogenous privacy $\lambda$ and $\beta$ are
pinned by primitives rather than by a single exogenous noise level.
\end{proof}

\section{Numerical illustration}
\label{app:numerics}

This appendix illustrates the Kyle-model subsidy of \Cref{cor:kyle-subsidy}
numerically. We first tabulate the equilibrium in dimensionless form, then map it
onto a BTC/USDT calibration, and close by describing the shape of the subsidy as a
function of the privacy-noise scale.

\subsection{Dimensionless equilibrium}

Normalize $\sigv=\sigu=1$ and report the equilibrium price impact $\lambda$, the
informed intensity $\beta$, and the magnitude of the maker's subsidy
$|\pi_M|=|\Pi_M|$ as functions of the privacy-noise scale $\sige$. By
\Cref{prop:kyle-eq,cor:kyle-subsidy} these are closed forms in
$(\sigv,\sigu,\sige)$; \Cref{tab:dimensionless} evaluates them.

\begin{table}[h]
\centering
\begin{tabular}{ccccc}
\toprule
$\sige$ & $\lambda$ & $\beta$ & $|\pi_M|$ \\
\midrule
$0$           & $0.500$ & $1.000$ & $0.000$ \\
$0.5$         & $0.447$ & $1.118$ & $0.112$ \\
$1$           & $0.354$ & $1.414$ & $0.354$ \\
$\sqrt2$      & $0.289$ & $1.732$ & $0.577$ \\
$2$           & $0.224$ & $2.236$ & $0.894$ \\
$3$           & $0.158$ & $3.162$ & $1.423$ \\
$5$           & $0.098$ & $5.099$ & $2.451$ \\
\bottomrule
\end{tabular}
\caption{Dimensionless Kyle equilibrium at $\sigv=\sigu=1$. As the privacy-noise
scale $\sige$ rises, price impact $\lambda$ falls, the informed trader trades more
aggressively ($\beta$ rises), and the privacy subsidy $|\pi_M|$ grows. The
no-privacy benchmark $\sige=0$ recovers textbook Kyle ($\lambda=\tfrac12$,
$\beta=1$) with zero subsidy.}
\label{tab:dimensionless}
\end{table}

\subsection{BTC/USDT calibration}

To put the subsidy in monetary units we calibrate the scales to a daily BTC/USDT
market: a value-uncertainty $\sigv=\$3{,}000$ per BTC (the daily standard deviation of
fundamental value) and a noise-flow $\sigu=1{,}000~\text{BTC}/\text{day}$. The product
$\sigv\sigu$ then carries units of USD/day, so $|\pi_M|$ is in USD/day and the ratio
$|\pi_M|/(\sigv\sigu)$ is dimensionless. \Cref{tab:btc} tabulates $|\pi_M|$ across noise
levels $\sige/\sigu$.

\begin{table}[h]
\centering
\begin{tabular}{ccc}
\toprule
$\sige/\sigu$ & $|\pi_M|$ (USD/day) & $|\pi_M|/(\sigv\sigu)$ \\
\midrule
$0.1$     & $\sim\$15\text{k}$    & $0.005$ \\
$0.5$     & $\sim\$335\text{k}$   & $0.112$ \\
$1$       & $\sim\$1.06\text{M}$  & $0.354$ \\
$\sqrt2$  & $\sim\$1.73\text{M}$  & $0.577$ \\
$2$       & $\sim\$2.68\text{M}$  & $0.894$ \\
\bottomrule
\end{tabular}
\caption{BTC/USDT calibration of the privacy subsidy at
$\sigv=\$3{,}000$ per BTC, $\sigu=1{,}000~\text{BTC}/\text{day}$. The fraction
column reproduces the dimensionless $|\pi_M|$ of \Cref{tab:dimensionless} and is
calibration-free.}
\label{tab:btc}
\end{table}

At the symmetric noise level $\sige=\sigu$ the subsidy is
\[
|\pi_M| \;=\; \frac{\sigv\,\sigu}{2\sqrt2}
\;=\; \frac{3000\cdot 1000}{2\sqrt2}
\;=\; \$1{,}060{,}660/\text{day}.
\]
The model-consistent way to read this as a fee is the \emph{rate} on its own cleared
notional. At a BTC price $p_0\approx\$60$k the calibration clears
$Q_0=p_0\,\E[|x|+|u|]\approx\$0.12\text{B}/\text{day}$ (informed plus liquidity flow),
so the break-even fee $\tau^\ast=|\pi_M|/Q_0$ is $\{1.6,33,92,132,173\}$ basis points
across $\sige/\sigu\in\{0.1,0.5,1,\sqrt2,2\}$. This rate is independent of the absolute
flow scale $\sigu$ (it depends only on $\sige/\sigu$ and $\sigv/p_0$), so although a
real BTC/USDT venue clears $\sim\$1\text{B}+/\text{day}$---about $10\times$ this
single-maker calibration---rescaling $\sigu$ to match leaves $\tau^\ast$ unchanged. The
point is the order of magnitude: light privacy is recoverable within a conventional
$\sim\!10$\,bp fee, while parity privacy demands roughly an order of magnitude more.

\subsection{Shape of the subsidy}

\Cref{fig:subsidy} plots $|\pi_M|(\sige)=\sige^2/\big(2\sqrt{1+\sige^2}\big)$ for
$\sigv=\sigu=1$, the dimensionless column of \Cref{tab:dimensionless} as a continuous
curve. The subsidy starts at zero in the no-privacy limit $\sige=0$, rises
monotonically in the noise scale, and is convex for small $\sige$ before bending toward
its asymptotically linear growth: the curve has an inflection point at
$\sige^\star=\sqrt2\,\sigu$ (here $\sqrt2$ at unit scale), the noise level at which the
marginal subsidy per unit of added privacy noise is maximal. This inflection concerns
the \emph{transfer} alone; the interior optimal noise scale of \Cref{thm:optsigma} is
instead driven by balancing the differential-privacy benefit's marginal value
$bc/\sige^2$ against the rising marginal deadweight $4A\sige^3$, a separate object from
the subsidy curve.

\begin{figure}[H]
\centering
\begin{tikzpicture}
\begin{axis}[
  width=0.72\textwidth, height=0.42\textwidth,
  xlabel={privacy noise scale $\sige$ (units of $\sigu$)},
  ylabel={subsidy $|\pi_M|$ (units of $\sigv\sigu$)},
  domain=0:5, samples=120, axis lines=left,
  xmin=0, ymin=0, legend pos=south east,
  every axis plot/.append style={thick},
]
\addplot[blue]{x^2/(2*sqrt(1+x^2))};
\draw[densely dotted] (axis cs:1.41421,0) -- (axis cs:1.41421,{1.41421^2/(2*sqrt(1+1.41421^2))});
\node[anchor=south] at (axis cs:1.41421,{1.41421^2/(2*sqrt(1+1.41421^2))}) {\scriptsize$\sige^\star=\sqrt2$};
\end{axis}
\end{tikzpicture}
\caption{The privacy subsidy $|\pi_M|=\sige^2/(2\sqrt{1+\sige^2})$ as a function of the
noise scale $\sige$ at $\sigv=\sigu=1$. Zero at $\sige=0$, monotonically increasing,
with an inflection at $\sige^\star=\sqrt2\,\sigu$ (dotted) separating the accelerating
from the decelerating regime.}
\label{fig:subsidy}
\end{figure}

\bibliographystyle{plainnat}
\bibliography{references}

\begin{thebibliography}{35}
\providecommand{\natexlab}[1]{#1}
\providecommand{\url}[1]{\texttt{#1}}
\expandafter\ifx\csname urlstyle\endcsname\relax
  \providecommand{\doi}[1]{doi: #1}\else
  \providecommand{\doi}{doi: \begingroup \urlstyle{rm}\Url}\fi

\bibitem[Aase and {\O}ksendal(2019)]{nonfiduciary-mm}
Knut~K. Aase and Bernt {\O}ksendal.
\newblock Strategic insider trading equilibrium with a non-fiduciary market
  maker.
\newblock \emph{arXiv preprint arXiv:1908.08777}, 2019.

\bibitem[Back(1992)]{back1992insider}
Kerry Back.
\newblock Insider trading in continuous time.
\newblock \emph{Review of Financial Studies}, 5\penalty0 (3):\penalty0
  387--409, 1992.

\bibitem[Bender and Kraut(2024)]{renegade-wp}
Christopher Bender and Joseph Kraut.
\newblock Renegade whitepaper, protocol specification v0.6.
\newblock \url{https://whitepaper.renegade.fi/}, 2024.
\newblock Accessed 2026-05-15.

\bibitem[Bergemann and Morris(2019)]{bergemann2019infodesign}
Dirk Bergemann and Stephen Morris.
\newblock Information design: A unified perspective.
\newblock \emph{Journal of Economic Literature}, 57\penalty0 (1):\penalty0
  44--95, 2019.

\bibitem[Brahma et~al.(2012)Brahma, Chakraborty, Das, Lavoie, and
  Magdon-Ismail]{brahma-bayesian-mm}
Aseem Brahma, Mithun Chakraborty, Sanmay Das, Allen Lavoie, and Malik
  Magdon-Ismail.
\newblock A {B}ayesian market maker.
\newblock In \emph{Proceedings of the 13th ACM Conference on Electronic
  Commerce (EC 2012)}, 2012.

\bibitem[Buti et~al.(2017)Buti, Rindi, and Werner]{buti2017darkpool}
Sabrina Buti, Barbara Rindi, and Ingrid~M. Werner.
\newblock Dark pool trading strategies, market quality and welfare.
\newblock \emph{Journal of Financial Economics}, 124\penalty0 (2):\penalty0
  244--265, 2017.

\bibitem[Caldentey and Stacchetti(2010)]{caldentey2010}
Ren{\'e} Caldentey and Ennio Stacchetti.
\newblock Insider trading with a random deadline.
\newblock \emph{Econometrica}, 78\penalty0 (1):\penalty0 245--283, 2010.

\bibitem[Carmier(2022)]{carmier2022thermo}
Pierre Carmier.
\newblock Generalized second law of thermodynamics in the {G}losten--{M}ilgrom
  model.
\newblock \emph{arXiv preprint arXiv:2209.15429}, 2022.

\bibitem[{\c C}etin and Danilova(2016)]{cetin2017financial}
Umut {\c C}etin and Albina Danilova.
\newblock Markovian nash equilibrium in financial markets with asymmetric
  information and related forward-backward systems.
\newblock \emph{Annals of Applied Probability}, 26\penalty0 (4):\penalty0
  1996--2029, 2016.

\bibitem[Chhaibi et~al.(2025)Chhaibi, Ekren, and Noh]{chhaibi2025solvability}
Reda Chhaibi, Ibrahim Ekren, and Eunjung Noh.
\newblock Solvability of the {G}aussian {K}yle model with imperfect information
  and risk aversion.
\newblock \emph{arXiv preprint arXiv:2501.16488}, 2025.

\bibitem[Chitra et~al.(2022)Chitra, Angeris, and Evans]{chitra2022dpcfmm}
Tarun Chitra, Guillermo Angeris, and Alex Evans.
\newblock Differential privacy in constant function market makers.
\newblock In \emph{Financial Cryptography and Data Security (FC 2022)}, Lecture
  Notes in Computer Science. Springer, 2022.
\newblock IACR eprint 2021/1101.

\bibitem[Danilova(2010)]{danilova2010imperfect}
Albina Danilova.
\newblock Stock market insider trading in continuous time with imperfect
  dynamic information.
\newblock \emph{Stochastics: An International Journal of Probability and
  Stochastic Processes}, 82\penalty0 (1):\penalty0 111--131, 2010.
\newblock arXiv:1607.00035.

\bibitem[Das(2005)]{das2005learning}
Sanmay Das.
\newblock A learning market-maker in the {G}losten--{M}ilgrom model.
\newblock \emph{Quantitative Finance}, 5\penalty0 (2):\penalty0 169--180, 2005.

\bibitem[Dwork and Roth(2014)]{dwork2014algorithmic}
Cynthia Dwork and Aaron Roth.
\newblock The algorithmic foundations of differential privacy.
\newblock \emph{Foundations and Trends in Theoretical Computer Science},
  9\penalty0 (3--4):\penalty0 211--407, 2014.

\bibitem[Easley and O'Hara(1992)]{easleyohara1992}
David Easley and Maureen O'Hara.
\newblock Time and the process of security price adjustment.
\newblock \emph{Journal of Finance}, 47\penalty0 (2):\penalty0 577--605, 1992.

\bibitem[Flashbots(2024)]{flashbots-suave}
Flashbots.
\newblock The future of {MEV} is {SUAVE}.
\newblock \url{https://writings.flashbots.net/the-future-of-mev-is-suave},
  2024.
\newblock Accessed 2026-05-15.

\bibitem[Foster and Viswanathan(1996)]{foster1996strategic}
F.~Douglas Foster and S.~Viswanathan.
\newblock Strategic trading when agents forecast the forecasts of others.
\newblock \emph{Journal of Finance}, 51\penalty0 (4):\penalty0 1437--1478,
  1996.

\bibitem[Glosten(1989)]{glosten1989}
Lawrence~R. Glosten.
\newblock Insider trading, liquidity, and the role of the monopolist
  specialist.
\newblock \emph{Journal of Business}, 62\penalty0 (2):\penalty0 211--235, 1989.

\bibitem[Glosten and Harris(1988)]{glosten1988}
Lawrence~R. Glosten and Lawrence~E. Harris.
\newblock Estimating the components of the bid/ask spread.
\newblock \emph{Journal of Financial Economics}, 21\penalty0 (1):\penalty0
  123--142, 1988.

\bibitem[Glosten and Milgrom(1985)]{glosten1985}
Lawrence~R. Glosten and Paul~R. Milgrom.
\newblock Bid, ask and transaction prices in a specialist market with
  heterogeneously informed traders.
\newblock \emph{Journal of Financial Economics}, 14\penalty0 (1):\penalty0
  71--100, 1985.

\bibitem[Huddart et~al.(2001)Huddart, Hughes, and
  Levine]{huddart2001disclosure}
Steven Huddart, John~S. Hughes, and Carolyn~B. Levine.
\newblock Public disclosure and dissimulation of insider trades.
\newblock \emph{Econometrica}, 69\penalty0 (3):\penalty0 665--681, 2001.

\bibitem[Kyle(1985)]{kyle1985}
Albert~S. Kyle.
\newblock Continuous auctions and insider trading.
\newblock \emph{Econometrica}, 53\penalty0 (6):\penalty0 1315--1335, 1985.

\bibitem[Milionis et~al.(2022)Milionis, Moallemi, Roughgarden, and
  Zhang]{milionis2022lvr}
Jason Milionis, Ciamac~C. Moallemi, Tim Roughgarden, and Anthony~Lee Zhang.
\newblock Automated market making and loss-versus-rebalancing.
\newblock \emph{arXiv preprint arXiv:2208.06046}, 2022.

\bibitem[Milionis et~al.(2024)Milionis, Moallemi, and
  Roughgarden]{moallemi2022myersonian}
Jason Milionis, Ciamac~C. Moallemi, and Tim Roughgarden.
\newblock A {M}yersonian framework for optimal liquidity provision in automated
  market makers.
\newblock In \emph{Innovations in Theoretical Computer Science (ITCS 2024)},
  2024.
\newblock arXiv:2303.00208.

\bibitem[Nakamura(2026{\natexlab{a}})]{nakamura2026continuoustime}
Yuki Nakamura.
\newblock The privacy subsidy in continuous-time {K}yle: Cumulative welfare
  under noise-perturbed order-flow observation, 2026{\natexlab{a}}.

\bibitem[Nakamura(2026{\natexlab{b}})]{nakamura2026glostenmilgrom}
Yuki Nakamura.
\newblock The privacy subsidy in {G}losten--{M}ilgrom: Bid-ask spread and
  welfare under flip-noise direction observation, 2026{\natexlab{b}}.

\bibitem[Nakamura(2026{\natexlab{c}})]{nakamura2026privacysubsidy}
Yuki Nakamura.
\newblock The privacy subsidy: Kyle's $\lambda$ under noise-perturbed
  order-flow observation, 2026{\natexlab{c}}.

\bibitem[{Penumbra Labs}(2024)]{penumbra-docs}
{Penumbra Labs}.
\newblock Penumbra protocol documentation.
\newblock \url{https://protocol.penumbra.zone/}, 2024.
\newblock Accessed 2026-05-15.

\bibitem[Qiu and Zhou(2023)]{qiu2023insider}
Jixiu Qiu and Yonghui Zhou.
\newblock Insider trading with dynamic asset under market makers' partial
  observations.
\newblock \emph{AIMS Mathematics}, 8\penalty0 (10):\penalty0 25017--25036,
  2023.
\newblock \doi{10.3934/math.20231277}.

\bibitem[Routledge et~al.(2025)Routledge, Shen, and
  Zetlin-Jones]{routledge2024amm}
Bryan~R. Routledge, Yikang Shen, and Ariel Zetlin-Jones.
\newblock Automated exchange economies.
\newblock Working paper, Tepper School of Business, 2025.

\bibitem[Touzo et~al.(2021)Touzo, Marsili, and Zagier]{touzo2020information}
L{\'e}o Touzo, Matteo Marsili, and Don Zagier.
\newblock Information thermodynamics of financial markets: the
  {G}losten--{M}ilgrom model.
\newblock \emph{Journal of Statistical Mechanics: Theory and Experiment}, 2021.
\newblock arXiv:2010.01905.

\bibitem[Viswanathan and Xing(2026)]{flexible-info-kyle}
S.~Viswanathan and Hao Xing.
\newblock Flexible information acquisition in the {K}yle model.
\newblock \emph{arXiv preprint arXiv:2603.21842}, 2026.

\bibitem[Warner(1965)]{warner1965rr}
Stanley~L. Warner.
\newblock Randomized response: A survey technique for eliminating evasive
  answer bias.
\newblock \emph{Journal of the American Statistical Association}, 60\penalty0
  (309):\penalty0 63--69, 1965.

\bibitem[Zhang et~al.(2025)Zhang, Li, Sun, Chen, and Chen]{zhang2025mevbatch}
Mengqian Zhang, Yuhao Li, Xinyuan Sun, Elynn Chen, and Xi~Chen.
\newblock Maximal extractable value in batch auctions.
\newblock In \emph{Proceedings of the 26th ACM Conference on Economics and
  Computation}, 2025.

\bibitem[Zhu(2014)]{zhu2014darkpools}
Haoxiang Zhu.
\newblock Do dark pools harm price discovery?
\newblock \emph{Review of Financial Studies}, 27\penalty0 (3):\penalty0
  747--789, 2014.

\end{thebibliography}

\end{document}